\documentclass[runningheads]{llncs}
\usepackage{subcaption}
\usepackage{graphicx}
\usepackage{amssymb}
\usepackage[ruled,vlined]{algorithm2e}

\begin{document}

\title{Minimum-link $C$-Oriented Paths Visiting a Sequence of Regions in the Plane}
\titlerunning{Minimum-link $C$-Oriented Paths}
%
\author{Kerem Geva\inst{1} \and
Matthew J. Katz\inst{1} \and
Joseph S. B. Mitchell\inst{2} \and
Eli Packer\inst{3}}

\authorrunning{K. Geva et al.}
%
\institute{Ben Gurion University of the Negev, Beer Sheva, Israel 
\email{\{gevak,matya\}@bgu.ac.il}\\
\and
Stony Brook University, Stony Brook, NY, USA
\email{joseph.mitchell@stonybrook.edu}\\
\and
Track160
\email{eli.p@track160.com}}
\maketitle              

\newtheorem{my_claim}{Claim}


\def\reals{{\mathbb R}}
\def\sphere{{\mathbb S}}
\def\eps{{\varepsilon}}
\def\bd{{\partial}}
\def\ttheta{{\bf \theta}}
\def\dist{\varrho}
\def\polylog{\mathop{\mathrm{polylog}}}
\def\mypara#1{\smallskip\noindent\textbf{\textit{#1}}}

\def\A{{\cal A}}
\def\D{{\cal D}}
\def\F{{\cal F}}
\def\G{{\cal G}}
\def\K{{\cal K}}
\def\L{{\cal L}}
\def\N{{\cal N}}
\def\R{{\cal R}}
\def\S{{\cal S}}
\def\T{{\cal T}}
\def\H{{\cal H}}
\def\N{{\cal N}}
\def\W{{\cal W}}

\def\etal{\textit{et~al.}}

\def\marrow{\marginpar[\hfill$\longrightarrow$]{$\longleftarrow$}}
\def\kerem#1{\textsc{(Kerem says: }\marrow\textsf{#1})}
\def\matya#1{\textsc{(Matya says: }\marrow\textsf{#1})}
\def\joe#1{\textsc{(Joe says: }\marrow\textsf{#1})}
\def\eli#1{\textsc{(Eli says: }\marrow\textsf{#1})}


\begin{abstract}
Let $E=\{e_1,\ldots,e_n\}$ be a set of $C$-oriented disjoint segments in the plane, where $C$ is a given finite set of orientations that spans the plane, and let $s$ and $t$ be two points. 
We seek a minimum-link $C$-oriented tour of $E$, that is, a polygonal path $\pi$ from $s$ to $t$ that visits the segments of $E$ in order, such that, the orientations of its edges are in $C$ and their number is minimum.
We present an algorithm for computing such a tour in $O(|C|^2 \cdot n^2)$ time. This problem already captures most of the difficulties occurring in the study of the more general problem, in which $E$ is a set of not-necessarily-disjoint $C$-oriented polygons.

\end{abstract}
\section{Introduction}
\label{sec:intro}

We consider the problem in which we are given a sequence of regions, ${\cal R}=(R_1,R_2,\ldots,R_n)$, where each $R_i$ is a subset of an underlying geometric domain, and our goal is to compute a tour (a path or a cycle) within the domain that visits the regions $\cal R$ in the given order and is optimal in some prescribed sense. Optimality might be based on the Euclidean length of the tour, the number of turns in a polygonal tour (or, equivalently, the number of {\em links} (edges) in the tour), a weighted cost function, etc. There are also variants of the problem in which it is important to specify exactly what constraints there are on the ordered visitation of the regions, particularly if the regions are not disjoint. The problem arises naturally and is also motivated by applications in curve simplification (e.g., \cite{DBLP:journals/ijcga/GuibasHMS93}), vehicle routing (e.g., the traveling salesperson problem (TSP); see \cite{crc-2016}), search and exploration (e.g., \cite{DBLP:conf/stoc/DrorELM03}), computing structures on imprecise points \cite{Loffler11}, task sequencing in robotics (see \cite{alatartsev2015robotic,alatartsev2013optimizing}), etc.

In this paper we focus on the version of the problem in which the regions $R_i$ are disjoint $C$-oriented line segments (with orientations/slopes from a finite set $C$) in the plane, the tour is required to be polygonal and $C$-oriented, and the optimality criterion is to minimize the number of links (equivalently, the number of turns, or vertices in the polygonal tour). We briefly mention generalizations (deferred to the full paper), including the case in which the regions $R_i$ are more general than disjoint line segments.

More formally, let $C$ be a finite set of orientations, which can be thought of as points on a unit circle centered at the origin. We assume that (i) $C$ spans the plane, i.e., for any two points $p,q$ in the plane, there exists a two-link (directed) path from $p$ to $q$ (or a one-link path), such that the orientation of the edges in the path belong to $C$, and (ii) for any orientation $c_i \in C$, the orientation $\overline{c_i}$ is also in $C$, where $\overline{c_i}$ is the opposite orientation of $c_i$.
The requirement for paths to be $C$-oriented arises in some settings (mechanical constraints) but also has advantages in lower/upper bounding of the turn angles, in comparison with polygonal paths having general links, which may form arbitrarily sharp turns.

We focus on the following problem: \emph{Minimum-link $C$-oriented tour of a sequence of $C$-oriented segments}:
Let $E = \{e_1,\ldots,e_n\}$ be a set of $C$-oriented disjoint segments, that is, if we think of $e \in E$ as a directed segment, by arbitrarily picking one of the two possible directions, then $e$'s orientation belongs to $C$. Let $s$ and $t$ be two points that do not belong to any of the segments in $E$. A \emph{tour} of $E$ is a polygonal path $\pi$ that begins at $s$ and ends at $t$ with the following property: There exists a sequence of points $p_1,\ldots,p_n$ on $\pi$, such that, $p_i$ precedes $p_{i+1}$, for $1 \le i \le n-1$, and $p_i \in e_i$, for $1 \le i \le n$. A tour is $C$-oriented, if the orientation of each of its edges belongs to $C$. 
We wish to compute a $C$-oriented minimum-link tour of $E$, that is, a $C$-oriented tour consisting of a minimum number of links (i.e., edges). 
    
Our main contribution is an efficient algorithm to compute a minimum-link $C$-oriented tour of a sequence of $n$ disjoint $C$-oriented line segments, in time $O(|C|^2 \cdot n^2)$. (The algorithm becomes $O(n)$ in the special case of $|C|=4$, e.g., axis-oriented paths.)

\subsection*{Related Work}
In the \emph{touring polygons problem} (TPP), one seeks a tour that is shortest in Euclidean length that visits a sequence of polygons; such a tour is found in polynomial time if the polygons are convex and is NP-hard in general (and has an FPTAS)~\cite{DBLP:conf/stoc/DrorELM03}. 
Minimization of the link length of a tour visiting a sequence of (possibly overlapping) disks is studied in \cite{DBLP:journals/ijcga/GuibasHMS93}, where the motivation for this ``ordered stabbing'' problem was curve and map simplification (see also \cite{neyer1999line}). In contrast with our problem specification, in \cite{DBLP:journals/ijcga/GuibasHMS93} the path edges are allowed to be of arbitrary orientation, not required to be $C$-oriented. This assumption leads to particular efficiency, as one can use an extension of linear-time line stabbing methods (see Egyed and Wenger~\cite{egyed1991ordered}) to execute a greedy algorithm efficiently. 
Computing a minimum-link $C$-oriented path from start to goal among obstacles has been studied as well, without requiring visitation of a sequence of regions; see \cite{mitchell2014minimum,speckmann2018homotopic}. 

\section{Preliminaries}
\label{sec:overview}

\noindent
{\bf Notation.}
For any $1 \leq i \leq n$, let 
$l(e_i)$ be the number of links in a minimum-link path that begins at $s$ and ends at a point on $e_i$. We only consider $C$-oriented paths to $e_i$ that visit the segments $e_1,\ldots,e_i$, as defined above. We refer to the number of links in such a path as its \emph{length}. We distinguish between paths to $e_i$ both by their length and by the orientation of their last link. Let $I(e_i,c_j)$ ($I^+(e_i,c_j)$) be the set of maximal intervals on $e_i$ formed by all paths of length $l(e_i)$ ($l(e_i)+1$) from $s$ to $e_i$, whose last link has orientation $c_j$. We set $I(e_i) = \bigcup_{c\in C} I(e_i,c)$ and $I^+(e_i) = \bigcup_{c\in C} I^+(e_i,c)$.

For an orientation $c_j \in C$, let $c_{j+1}$ and $c_{j-1}$ be the orientations in $C$ that immediately succeed $c_j$ and precede $c_j$ in clockwise order, respectively. We denote by $\phi(c_j,c_k)$ the set of orientations in $C$ between $c_j$ and $c_k$ (in clockwise order from $c_j$), not including $c_j$ and $c_k$.   
Finally, we denote the ray emanating from $p$ in orientation $c_j$ by $Ray(p,c_j)$ and the line through $p$ parallel to a segment of orientation $c_j$ by $Line(p,c_j)$.

Let $a$ be an interval on $e_i$ that belongs to one of the sets $I(e_i)$ or $I^+(e_i)$. Then $a$ has a length $l_a$ (which is either $l(e_i)$ or $l(e_i)+1$) and an orientation $c_a \in C$ associated with it. We denote the endpoints of $a$ by $a_1$ and $a_2$, where $a_1$ is to the left of $a_2$, when approaching $a$ through a path corresponding to $a$ (i.e., a path starting at $s$ and ending at a point in $a$, which is of length $l_a$ and whose last link is of orientation $c_a$). Next, we use $a$ to define two regions of the plane, namely, $PT(a)$ and $\psi(a,c_j)$.

Let $PT(a)$ denote the semi-slab consisting of all points that can be reached by extending the last link of a path corresponding to $a$. We refer to such a path as a path that \emph{passes through $a$} and continues in the same orientation at which it reached $a$ (i.e., $c_a$). Thus, the region $PT(a)$ is the semi-slab bounded by the rays $Ray(a_1,c_a), Ray(a_2,c_a)$ and the interval $a$ (see, e.g., the red region in Figure~\ref{fig:segments_intersect}). 
Similarly, let $\psi(a,c_j)$ be the region of all points that can be reached by a path that passes through $a$ and then, not necessarily immediately,
turns and continues in orientation $c_j$. Thus, $\psi(a,c_j) = \bigcup_{q\in PT(a)} Ray(q,c_j)$, for example if $c_j = \overline{c_a}$, then $\psi(a,c_j)$ is the slab defined by the lines $Line(a_1,c_a)$ and $Line(a_2,c_a)$ (for additional examples see Figure~\ref{fig:delta_region}). 

Finally, for an interval $b \in I^+(e_i)$, we set $\delta(b) = \{a \in I(e_i) | a \subseteq b \}$.

We now show that the sets $I(e_i)$ and $I^+(e_i)$ are sufficient, in the sense that there exists a minimum-link tour of $E$ whose portion from $s$ to $e_i$ corresponds to an interval in $I(e_i) \cup I^+(e_i)$. Assume this is false, and let $\pi$ be a minimum-link tour of $E$, such that its portion $\pi_i$ from $s$ to $e_i$ does not correspond to an interval in $I(e_i) \cup I^+(e_i)$. Then, the length of $\pi_i$ (denoted $|\pi_i|$) is at least $l(e_i)+2$. Let $p$ be the point on $e_i$ where $\pi_i$ ends, and denote the portion of $\pi$ from $p$ to $t$ by $\pi^i$. Then $|\pi| \ge l(e_i)+2+|\pi^i|$, if $\pi$ makes a turn at $p$, or $|\pi| = l(e_i)+2+|\pi^i|-1$, otherwise.  Consider any path $\pi'_i$ from $s$ to $e_i$ that corresponds to an interval in $I(e_i)$ and let $p'$ be the point on $e_i$ where $\pi'_i$ ends. Then, the tour obtained by $\pi'_i$, the edge $p'p$ and $\pi^i$ is a tour of $E$ of length at most $l(e_i)+1+|\pi^i| \le |\pi|$.
We have thus shown that

\begin{my_claim}
\label{claim:claim1}
There exists a minimum-link tour of $E$ whose portion from $s$ to $e_i$ corresponds to an interval in $I(e_i) \cup I^+(e_i)$, for $1 \le i \le n$.
\end{my_claim}

Finally, since our assumptions on the set of orientations $C$ imply that there exists a two-link path from $p$ to $q$, for any pair of points $p,q$ in the plane, we have
\begin{my_claim}
\label{claim:claim2}
$l(e_{i-1}) \le l(e_i) \le l(e_{i-1}) +2$, for $1 \le i \le n$ (where $l(e_0)=0$).
\end{my_claim}

\section{The Main Algorithm}
\label{sec:algorithm}

In this section, we present an algorithm for computing a minimum-link tour of $E$. The algorithm consists of two stages. In the first stage, it considers the segments of $E$, one at a time, beginning with $e_1$, and, at the current segment $e_i$, it computes the sets $I(e_i)$ and $I^+(e_i)$ from the sets $I(e_{i-1})$ and $I^+(e_{i-1})$, associated with the previous segment. In the second stage, it constructs a minimum-link tour of $E$, beginning from its last link, by consulting the sets $I(\cdot)$ and $I^+(\cdot)$ computed in the first stage.

We begin with several definitions that will assist us in the description of the algorithm.
Given a set $I$ of intervals on $e_i$, where each interval $a \in I$ is associated with some fixed length (link distance) $l_a=l$ and an orientation $c_a$, and $c_j \in C$, we define the sets of intervals \emph{+0-intervals, +1-intervals, +2-intervals} on $e_{i+1}$ with respect to $I$ and $c_j$ (the definition of the first set does not depend on $c_j$).

The {\bf +0-intervals} on $e_{i+1}$ consist of the intervals on $e_{i+1}$ formed by passing through the intervals of $I$, without making any turns. It is constructed by computing the interval $b = PT(a) \cap e_{i+1}$, for each $a \in I$, and including it in the set, setting $l_b = l$ and $c_b = c_a$, if it is not empty.

The {\bf +1-intervals} on $e_{i+1}$ associated with orientation $c_j$ consist of the intervals on $e_{i+1}$ formed by passing through the intervals of $I$ and then making a turn in orientation $c_j$. It is constructed by computing the interval $b = \psi(a,c_j) \cap e_{i+1}$, for each $a \in I$, and including it in the set, setting $l_b = l+1$ and $c_b = c_j$.

The {\bf +2-intervals} on $e_{i+1}$ associated with orientation $c_j$ consist of the intervals on $e_{i+1}$ formed by passing through the intervals of $I$ and then making two turns, where the first is in any orientation $c \neq \overline{c_a}$ and the second is in orientation $c_j$; see Lemma~\ref{claim:lemma4.1}.

We construct it as follows.
First, we check if there is an interval $a \in I$ such that $c_a \notin \{c_{j-1},c_j,c_{j+1}\}$. If there is such an interval, we include the interval $b = e_{i+1}$, setting $l_b = l+2$ and $c_b = c_j$, and stop; see Lemma~\ref{claim:lemma4.2}.
Otherwise, for each $a\in I$, we include the intervals $b^+ = \psi(a,\overline{c_a}_{+1}) \cap e_{i+1}$ and $b^- = \psi(a,\overline{c_a}_{-1}) \cap e_{i+1}$, provided that they are not empty, and set $l_{b^+} = l_{b^-} = l+2$ and $c_{b^+} = c_{b^-} = c_j$; see paragraph following Lemma~\ref{claim:lemma4.2}.

\subsection{Stage I}
We are now ready to describe the first stage of the algorithm. It is convenient to treat the points $s$ and $t$ as segments $e_0$ and $e_{n+1}$, respectively. We set $l(e_0)=0$ and, for each $c_j \in C$,  we insert the interval $a=e_0$, after setting $l_a=0$ and $c_a = c_j$, into $I(e_0,c_j)$. Similarly, for each $c_j \in C$, we insert the interval $a=e_0$, after setting $l_a=1$ and $c_a = c_j$, into $I^+(e_0,c_j)$.

We iterate over the segments $e_1,\ldots,e_{n+1}$, where in the $i$'th iteration, $1 \le i \le n+1$, we compute $l(e_i)$ and the pair of sets $I(e_i)$ and $I^+(e_i)$. Assume we have already processed the segments $e_0,\ldots,e_i$, for some $0 \le i \le n$. We describe the next iteration, in which we compute $l(e_{i+1})$ and the sets $I(e_{i+1})$ and $I^+(e_{i+1})$.

For each $c_j \in C$, we compute the +0-intervals on $e_{i+1}$ with respect to $I(e_i,c_j)$ and store them in $I(e_{i+1},c_j)$. If at least one of the sets $I(e_{i+1},c_j)$ is non-empty, we set $l(e_{i+1}) = l(e_i)$ (otherwise $l(e_{i+1}) > l(e_i)$). 
Next, for each $c_j \in C$, we compute the +0-intervals on $e_{i+1}$ with respect to $I^+(e_i,c_j)$ and the +1-intervals on $e_{i+1}$ with respect to $I(e_i) \setminus I(e_i,c_j)$ (and $c_j$). We store these intervals (if exist) either in $I^+(e_{i+1},c_j)$, if $l(e_{i+1})=l(e_i)$, or in $I(e_{i+1},c_j)$, if $l(e_{i+1}) > l(e_i)$. If we performed the latter option, then we set $l(e_{i+1}) = l(e_i)+1$. Finally, if we performed one of the two options, then we repeatedly merge overlapping intervals in the set (either $I^+(e_{i+1},c_j)$ or $I(e_{i+1},c_j)$), until there are no such intervals.

If $l(e_{i+1}) > l(e_i)$, then, for each $c_j \in C$, we compute the +2-intervals on $e_{i+1}$ with respect to $I(e_i)$ and the +1-intervals on $e_{i+1}$ with respect $I^+(e_i) \setminus I^+(e_i,c_j)$. We store these intervals (if exist) either in $I^+(e_{i+1},c_j)$, if $l(e_{i+1})=l(e_i)+1$, or in $I(e_{i+1},c_j)$, otherwise (i.e., we still have not fixed $l(e_{i+1})$). If we performed the latter option, then we set $l(e_{i+1}) = l(e_i)+2$, and, as above, if we performed one of the two options, then we repeatedly merge overlapping intervals in the set, until there are no such intervals. 

Finally, if $l(e_{i+1})=l(e_i)+2$, then, for each $c_j \in C$, we set $I^+(e_{i+1},c_j) = e_{i+1}$; see Claim~\ref{claim:claim5}.

\subsection{Stage II}
In this stage we use the information collected in the first stage to construct a minimum-link tour $\pi$ of $E$.

We construct $\pi$ incrementally beginning at $t$ and ending at $s$. That is, in the first iteration we add the portion of $\pi$ from $t$ to $e_n$, in the second iteration we add the portion from $e_n$ to $e_{n-1}$, etc. 
Assume that we have already constructed the portion of $\pi$ from $t$ to $e_i$, where this portion ends at point $p$ of interval $a$ on $e_i$. 
We describe in Algorithm~\ref{alg:Recovery} (see Appendix~\ref{app:stage2})
how to compute the portion from $e_i$ to $e_{i-1}$, which begins at the point $p$ of interval $a$ and ends at a point $p'$ of interval $b$ on $e_{i-1}$ (where $b \in I(e_{i-1}) \cup I^+(e_{i-1})$) and consists of $l_a-l_b+1$ links. 
Before continuing to the next iteration, we set $p=p'$ and $a=b$.

After adding the last portion, which ends at $s$, we remove all the redundant vertices from $\pi$, i.e., vertices at which $\pi$ does not make a turn.

\section{Analysis}

In this section, we prove the correctness of our two-stage algorithm and bound its running time, via a sequence of lemmas and claims.

\label{sec:analysis}
\begin{lemma}
\label{claim:lemma3.1}
For any interval $a \in I(e_i)$ and for any $c_j \in C \setminus \{c_a\}$, there exists an interval $b \in I^+(e_i,c_j)$ such that $a \subseteq b$, for $1\leq i \leq n$. 
\end{lemma}

\begin{proof}
Let $p\in a$, then there is a path $\pi_i$ of length $l(e_i)$ that begins at $s$, ends at $p$, and whose last link is of orientation $c_a$.
By making a turn at $p$ in orientation $c_j$ (without extending $\pi_i$), we obtain a path $\pi'_i$ of length $l(e_i)+1$, whose last link is of orientation $c_j$.
Therefore, there is an interval $b\in I^+(e_i,c_j)$ such that $p\in b$, and since (by construction) there are no overlapping intervals in $I^+(e_i,c_j)$, we conclude that $a \subseteq b$.
\end{proof}

\begin{lemma}
\label{lem:lemma3.2}
For any $1\leq i \leq n-1$ and $c_j\in C$, if there is an interval $a\in I(e_i,\overline{c_j}) \cup I^+(e_i,\overline{c_j})$ such that $PT(a)\cap e_{i+1}\neq \emptyset$, then, for any interval $b\in I(e_i,c_j) \cup I^+(e_i,c_j)$, we have that $PT(b)\cap e_{i+1} = \emptyset$.
\end{lemma}

\begin{proof}
If there exist intervals $a \in I(e_i,\overline{c_j}) \cup I^+(e_i,\overline{c_j})$ and $b \in I(e_i,c_j) \cup I^+(e_i,c_j)$, such that $e_{i+1}$ intersects both $PT(a)$ and $PT(b)$, then $e_{i+1}$ must intersect $e_i$ (see Figure~\ref{fig:segments_intersect}) --- contradiction. 
\end{proof}

The following claim bounds the number of intervals with associated length and orientation $l(e_i)+1$ and $c_j$, respectively, that are `created' on $e_{i+1}$.

\begin{my_claim}
\label{claim:claim3}
At most $\max\left\{ |I(e_i,\overline{c_j})|,|I^+(e_i,c_j)| \right\}  + 2$ intervals with associated length and orientation $l(e_i)+1$ and $c_j$, respectively, are `created' on $e_{i+1}$, during the execution of the algorithm.
\end{my_claim}

\begin{proof}
There are two ways to reach a point on $e_{i+1}$ with a path of length $l(e_i)+1$ whose last link is of orientation $c_j$. The first is by passing through one of the intervals in $I(e_i) \setminus I(e_i,c_j)$ and then making a turn in orientation $c_j$.
The second is by passing through one of the intervals in $I^+(e_i,c_j)$, without making any turn.
That is, the intervals on $e_{i+1}$ with associated length $l(e_i)+1$ and associated orientation $c_j$ are determined by the intervals in $I^+(e_i,c_j) \cup (I(e_i)\setminus I(e_i,c_j))$. 

Consider an interval $b \in I^+(e_i,c_j)$ (e.g., the blue interval in Figure~\ref{fig:first_step}), and let $c\in C$ be the orientation of $e_i$ when directed from $b_1$ to $b_2$. We divide $\delta(b)\setminus I(e_i,\overline{c_j})$ into four subsets as follows:
$A=\{a\in \delta(b) \mid c_a\in \phi(c_j,c) \cup \{c\}\}$,
$B=\{a\in \delta(b) \mid c_a\in \phi(c,\overline{c_j})\}$,
$C=\{a\in \delta(b) \mid c_a\in \phi(\overline{c_j}, \overline{c})\}$, and
$D=\{a\in \delta(b)\mid c_a\in \phi(\overline{c}, c_j)\cup \{\overline{c}\}\}$.
We denote by $R_{b \cup A}$ the region of all points that can be reached by a path that passes through $b$, or passes through $a\in A$ and then makes a turn in orientation $c_j$ (i.e., $R_{b \cup A}=PT(b) \cup \bigcup_{a \in A} \psi(a,c_j)$).
We compute the boundary of $R_{b \cup A}$ from $PT(b)$, by adding the regions $\psi(a,c_j)$, one at a time, for each interval $a \in A$. 

Let $\psi(a,c_j)$, for some $a \in A$, be the region that is added in the first step (see the red interval in Figure~\ref{fig:first_step}).
Since $(a_1,a_2)\subseteq (b_1,b_2)$ and $c_a \in \phi(c_j,c) \cup \{c\}$, $Ray(a_2,c_a)$ and $Ray(b_2,c_j)$ intersect at a point $p_a$ (see Figure~\ref{fig:first_step}).
By passing through $a$ and then turning before reaching $Ray(b_2,c_j)$ (i.e., at one of the points belonging to $PT(a) \cap PT(b)$), we cannot reach any point that is not already in $PT(b)$. However, by turning after crossing $Ray(b_2,c_j)$, we can reach points that are in the area bounded by $Ray(p_a,c_j)$ and $Ray(p_a,c_a)$ (the shaded area in Figure~\ref{fig:first_step}).
Thus, the region $R_{b \cup A}$ at the end of the first step, is bounded by $Ray(b_1,c_j)$, $(b_1,b_2)$, $(b_2,p_a)$ and $Ray(p_a,c_a)$, as can be seen in Figure~\ref{fig:first_step}.
\begin{figure}[!ht]
\hspace{-4em}%
    \begin{minipage}{.6\textwidth}
    \centering
    \vspace{10pt}
    \includegraphics[scale=0.5]{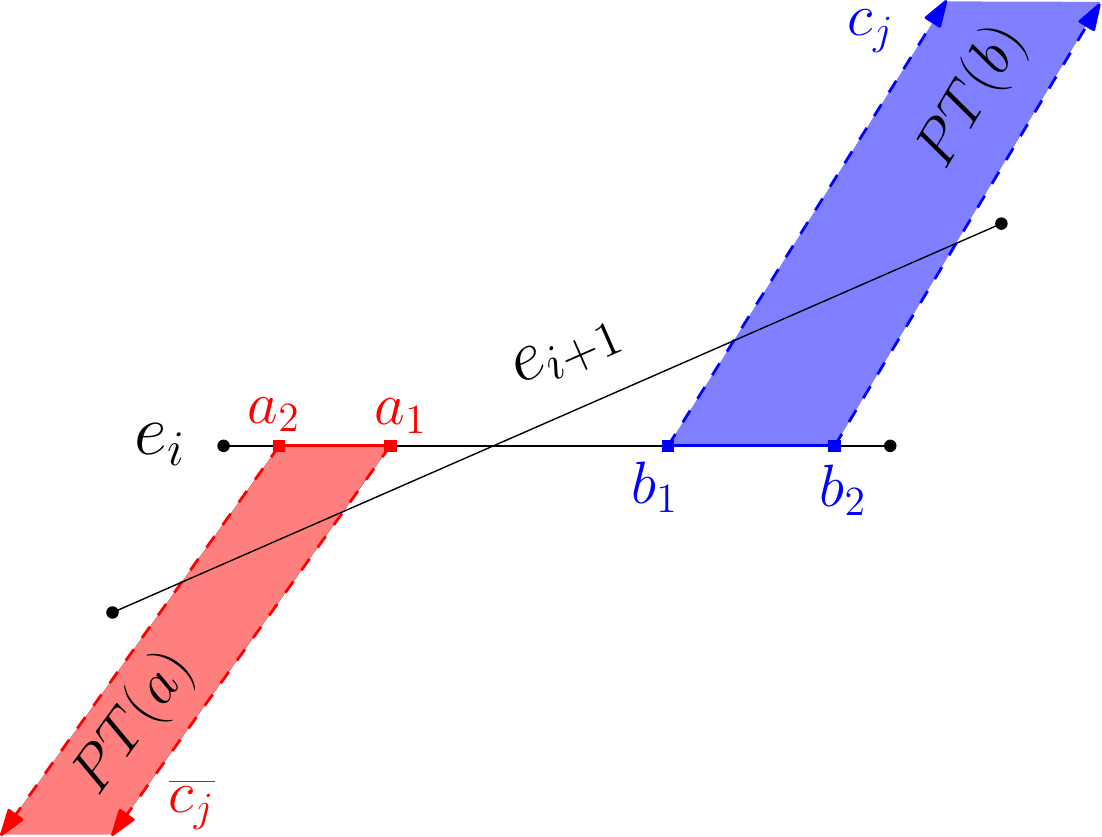}
    \caption{
    If $e_{i+1}$ intersects both $PT(a)$ (red) and $PT(b)$ (blue), then $e_{i+1}$ must intersect $e_i$.}
    \label{fig:segments_intersect}
    \end{minipage}%
     \hspace{5em}%
    \begin{minipage}{.4\textwidth}
    \centering
    \includegraphics[scale=0.5]{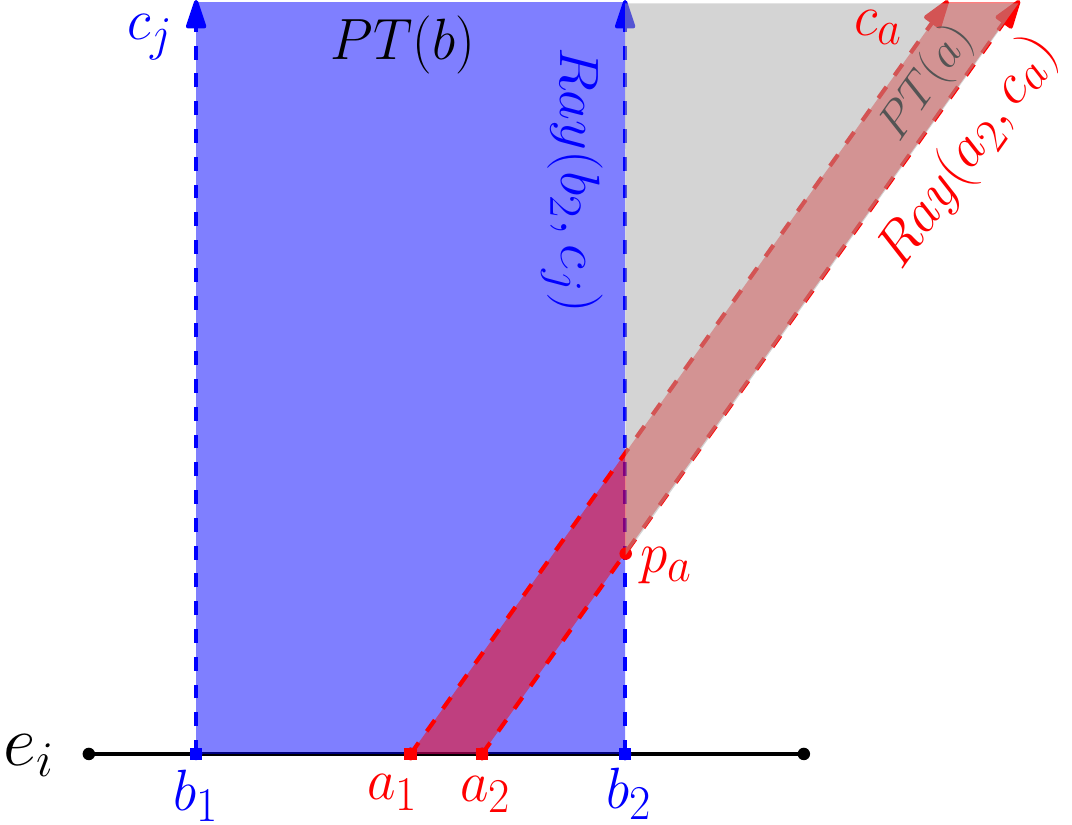}
      \newline
    \caption{The region $R_{b \cup A}$ at the end of the first step.}
    \label{fig:first_step}
    \end{minipage}%
\end{figure}
Notice the semi-infinite convex 2-chain that we obtain at the end of the first step, namely, the chain consisting of $(b_2,p_a)$ followed by  $Ray(p_a,c_a)$. It is easy to see that the region $R_{b \cup A}$ at the end of the last step, is bounded by $Ray(b_1,c_j)$, $(b_1,b_2)$, and a semi-infinite convex chain, denoted $l_A$, consisting of at most $|A|+1$ edges (see red chain in Figure~\ref{fig:R_boundaries_c}).
Finally, if $A=\emptyset$, then $R_{b\cup A}=PT(b)$ and we set $l_A=Ray(b_2,c_j)$.

Next, we set $R_{b \cup D}=PT(b) \cup \bigcup_{a \in D} \psi(a,c_j)$, and compute the convex chain $l_D$, which, together with $Ray(b_2,c_j)$ and $(b_1,b_2)$, defines the boundary of $R_{b\cup D}$.
Once again, if $D=\emptyset$, we set $l_D=ray(b_1,c_j)$.

Finally, we compute in a similar manner the convex chains $l_B$, which defines (together with $Ray(b_1,c_j)$) the boundary of $R_{b\cup B}=PT(b) \cup \bigcup_{a \in B} \psi(a,c_j)$ (see  purple chain in Figure~\ref{fig:R_boundaries_b}), and  $l_C$, which defines (together with $Ray(b_2,c_j)$) the boundary of $R_{b\cup C}=PT(b) \cup \bigcup_{a \in C} \psi(a,c_j)$. 

We now set $R= R_{b\cup A} \cup R_{b\cup B}\cup R_{b\cup C}\cup R_{b\cup D}$, then $R$ is the region of all points that can be reached by a path that passes through $b$, or passes through $a \in \delta(b) \setminus I(e_i,\overline{c_j})$ and then makes a turn in orientation $c_j$. Therefore, $R\cap e_{i+1}$ gives us the intervals on $e_{i+1}$ with length $l(e_i)+1$ and orientation $c_j$, which are created by passing through an interval in $\{b\} \cup \delta(b) \setminus  I(e_i,\overline{c_j})$.

In order to find these intervals, we identify the boundary of $R$ in each of the following four cases:
\begin{itemize}
\item {\bf Case A:} $B=\emptyset$ and $C=\emptyset$ (as illustrated in Figure~\ref{fig:R_boundaries_a})
\newline
In this case,
$R= R_{b\cup A}\cup R_{b\cup D}$, since $R_{b\cup B}=R_{b\cup C}=PT(b)$ and $PT(b)\subseteq R_{b\cup A}, R_{b\cup D}$, and $R$'s boundary is composed of $l_A$, $l_D$ and $b$.

\item {\bf Case B:} $B\neq \emptyset$ and $C= \emptyset$ (as illustrated in Figure~\ref{fig:R_boundaries_b})
\newline
In this case, the boundary of $R$ is composed of $l_B$ and $l_D$, since $R_{b\cup A} \subseteq R_{b\cup B}$.

\item {\bf Case C:}
$B=\emptyset$ and $C\ne\emptyset$ (as illustrated in Figure~\ref{fig:R_boundaries_c})
\newline
In this case, the boundary of $R$ is composed of $l_A$ and $l_C$, since $R_{b\cup D} \subseteq R_{b\cup C}$.
\item {\bf Case D:}
$B\ne\emptyset$ and $C\ne\emptyset$ (as illustrated in Figure~\ref{fig:R_boundaries_d})
\newline
in this case, $R=R_{b\cup B} \cup R_{b\cup C}$, and its boundary is the convex chain $l$ that is obtained from the chains $l_B$ and $l_C$, see Figure~\ref{fig:R_boundaries_d}.
\end{itemize}

\begin{figure}[!h]
\begin{subfigure}{.55\textwidth}
  \includegraphics[scale=0.40]{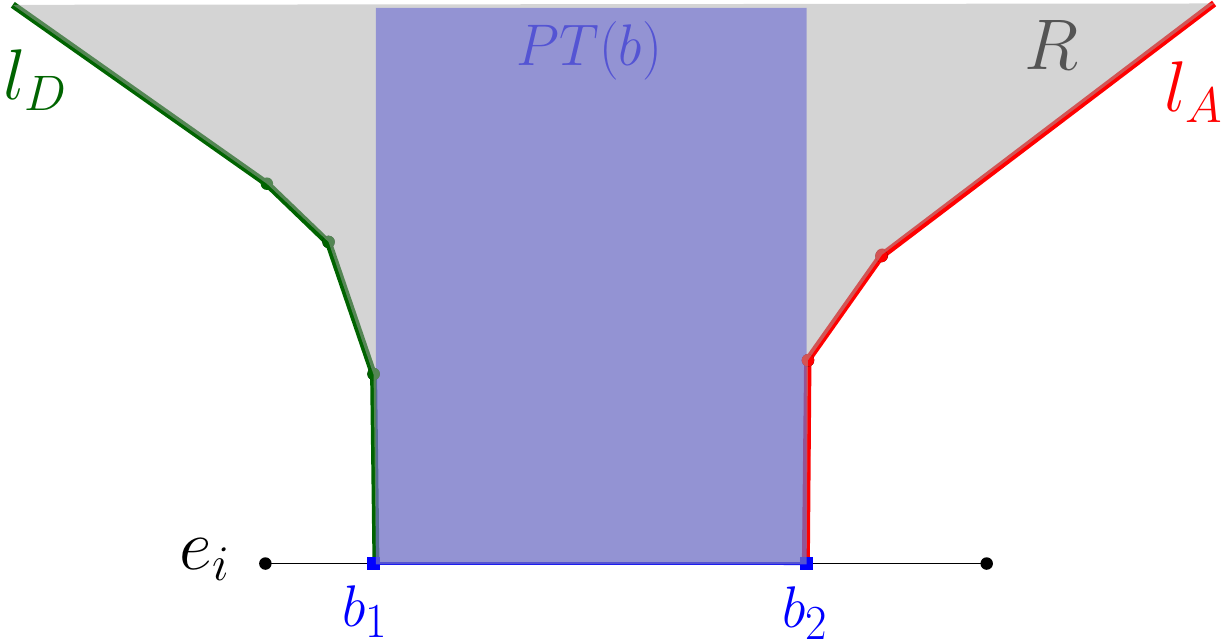}
  \newline
  \newline
  \newline
  \newline
  \caption{Case A: If $B=\emptyset$ and $C=\emptyset$, then $R$ is bounded by $l_A$ (red), $l_D$ (green) and $b$.}
  \label{fig:R_boundaries_a}
\end{subfigure}%
 \hspace{2em}%
      \begin{subfigure}{.55\textwidth}
  \includegraphics[scale=0.35]{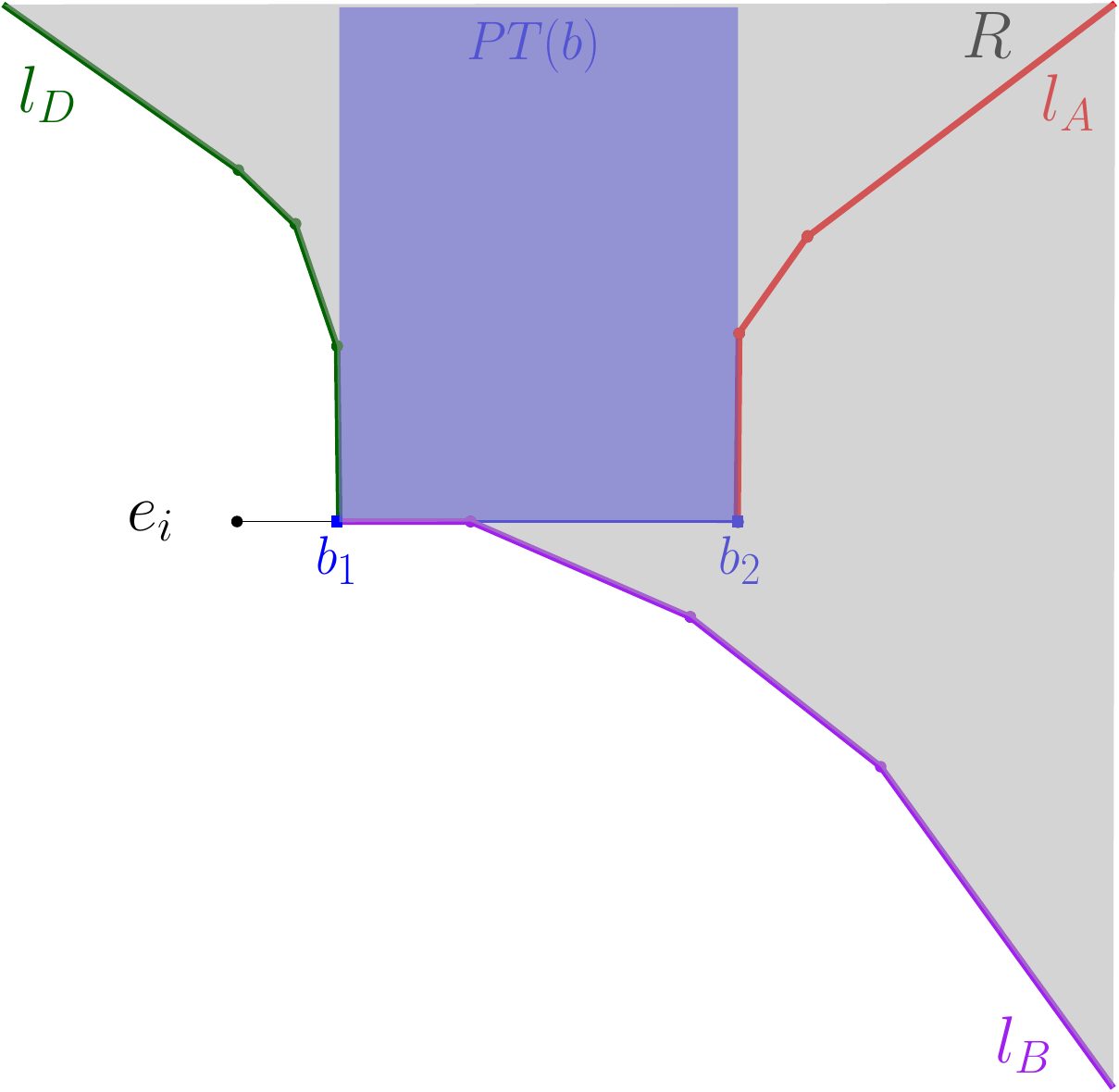}
  \caption{Case B: If $B\neq \emptyset$ and $C= \emptyset$, then $R$ is bounded by $l_B$ (purple) and $l_D$ (green).}
  \label{fig:R_boundaries_b}
  \end{subfigure}%
  \newline
    \newline
  \begin{subfigure}{.55\textwidth}
  \hspace{-3em}%
  \centering
  \includegraphics[scale=0.35]{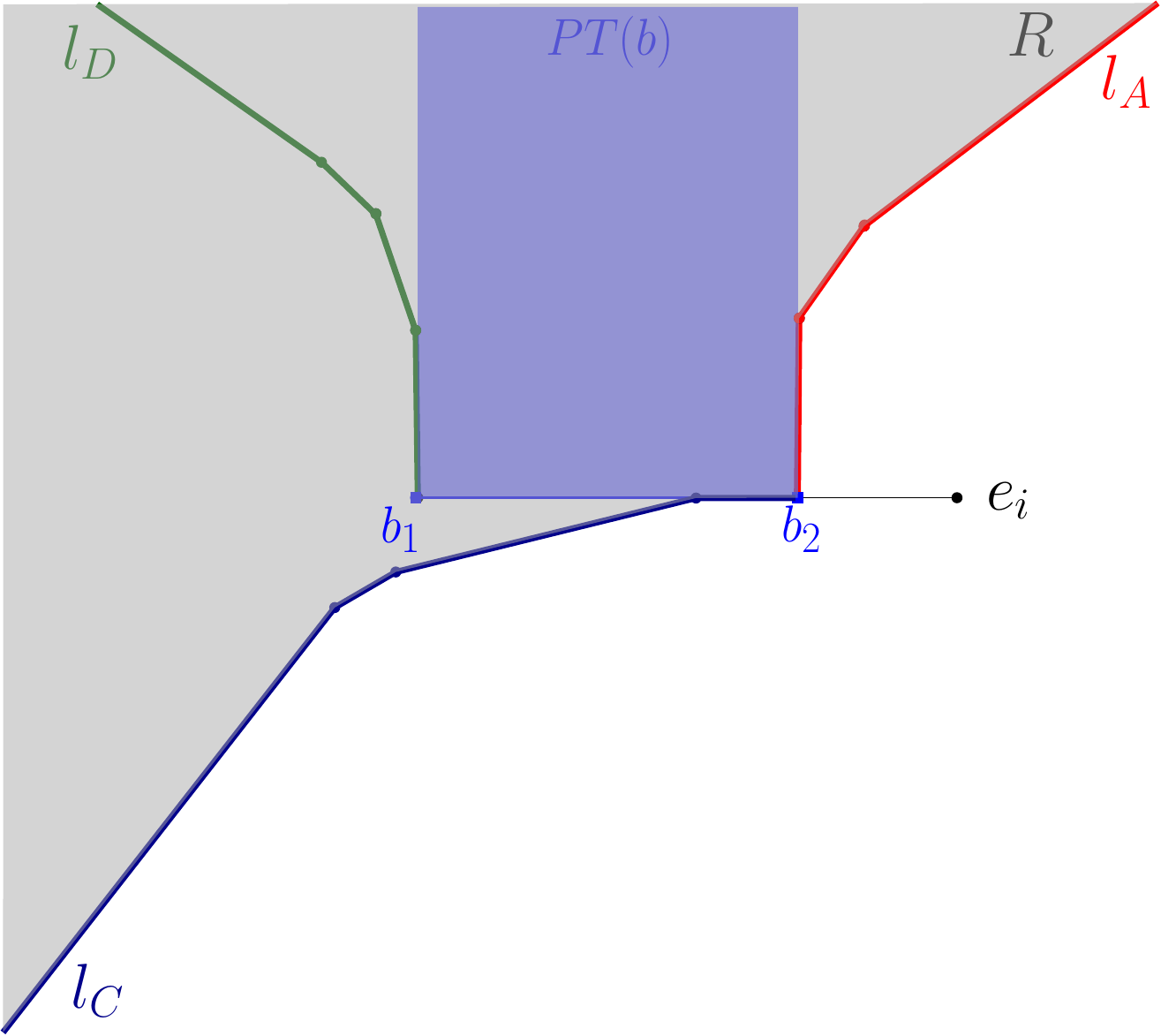}
  \caption{Case C: If $B= \emptyset$ and $C\neq \emptyset$, then $R$ is bounded by  $l_A$ (red) and $l_C$ (dark blue).}
  \label{fig:R_boundaries_c}
\end{subfigure}%
 \hspace{2em}%
      \begin{subfigure}{.55\textwidth}
  \includegraphics[scale=0.35]{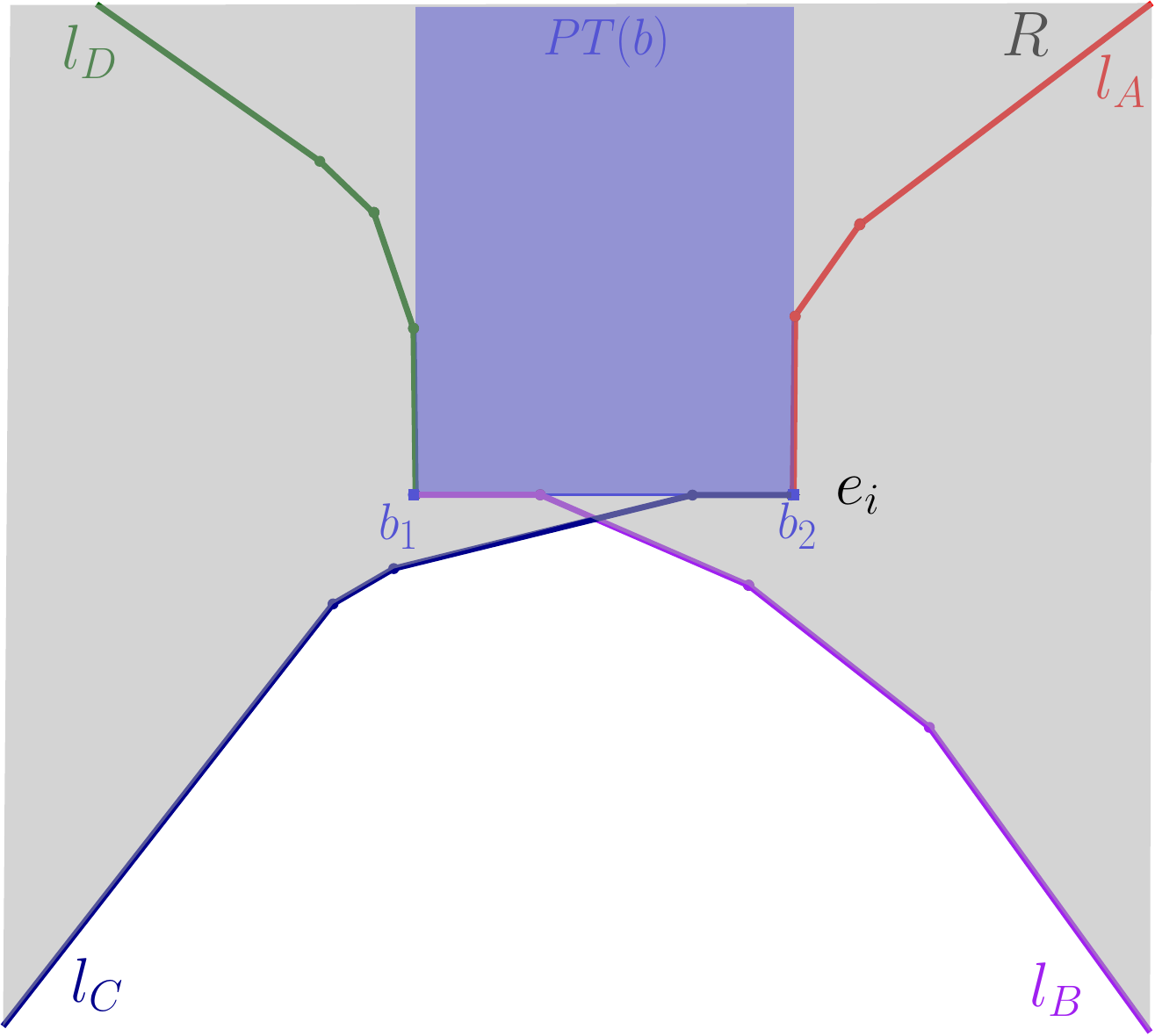}
  \caption{Case D: If $B\neq \emptyset$ and $C\neq \emptyset$, then $R$ is bounded by the chain $l$ obtained from  $l_B$ (purple) and $l_C$ (dark blue).}
  \label{fig:R_boundaries_d}
  \end{subfigure}%
\caption{The boundary of $R$.}
\label{fig:R_boundaries}
\end{figure}

We now examine how $e_{i+1}$ can intersect $R$, in each of these cases.
First, if $e_{i+1}$ does not intersect the boundary of $R$, then either $R\cap e_{i+1} = e_{i+1}$ or $R\cap e_{i+1} = \emptyset$.
In the former case, one interval is formed on $e_{i+1}$, which contains both its endpoints, and in the latter case, no interval is formed on $e_{i+1}$.
Next, assume that $e_{i+1}$ intersects the boundary of $R$. We distinguish between the case where there is an interval $h\in I(e_i,\overline{c_j}$) such that $PT(h)\cap e_{i+1}\neq \emptyset$, and the case where there is no such interval.

\vspace{2mm}
\noindent
{\bf There is an interval $h \in I(e_i,\overline{c_j}$) such that $PT(h)\cap e_{i+1}\neq \emptyset$.}\\
Then, by Lemma~\ref{lem:lemma3.2}, $PT(b)\cap e_{i+1}= \emptyset$.

{\bf If Case A:}
Clearly, $e_{i+1}$ cannot intersect both $l_A$ and $l_D$, since this would imply $PT(b)\cap e_{i+1}\neq \emptyset$ (see Figure~\ref{fig:intersection_R_a}).
Therefore, $e_{i+1}$ intersects exactly one of these chains, either at a single point or at two points. If $e_{i+1}$ intersects the chain at a single point $q$, then a single interval is formed on $e_{i+1}$, whose endpoints are $q$ and the endpoint of $e_{i+1}$ that lies in $R$ (see the edge $e^1_{i+1}$ in Figure~\ref{fig:e_intersect_R_a}).
If $e_{i+1}$ intersects the chain at two points, $p$ and $p'$, then two intervals are formed on $e_{i+1}$. The endpoints of these intervals are $p$ and $p'$ on one side and the corresponding endpoints of $e_{i+1}$ on the other side (see the edge $e^2_{i+1}$ in  Figure~\ref{fig:e_intersect_R_a}).

{\bf If Case B:}
Unlike Case~A, the fact that $PT(b) \cap e_{i+1} = \emptyset$ does not prevent $e_{i+1}$ from intersecting both $l_B$ and $l_D$. However, $e_{i+1}$ can intersect these chains in at most two points (in total), and as in Case~A at most two intervals are formed on $e_{i+1}$, where each of them contains an endpoint of $e_{i+1}$ (see Figure~\ref{fig:intersection_R_b}).

{\bf If Case C:}
Since Cases~B and~C are symmetric, at most two intervals are formed on $e_{i+1}$, each of which contains an endpoint of $e_{i+1}$.

{\bf If Case D:}
If $e_{i+1}$ intersects $l$ at a single point $q$, then a single interval is formed on $e_{i+1}$, whose endpoints are $q$ and the endpoint of $e_{i+1}$ that lies in $R$.
If $e_{i+1}$ intersects $l$ at two points $p$ and $p'$,
then $R\cap e_{i+1}$ consist of all the points on $e_{i+1}$, except for those in the interior of $(p,p')$. Therefore, two intervals are formed on $e_{i+1}$, and their endpoints are $p$ and $p'$ on one side and the corresponding endpoints of $e_{i+1}$ on the other side (see  Figure~\ref{fig:e_intersect_R_b}) 

\begin{figure}[h!]
\begin{subfigure}{.45\textwidth}
\centering
  \includegraphics[scale=0.35]{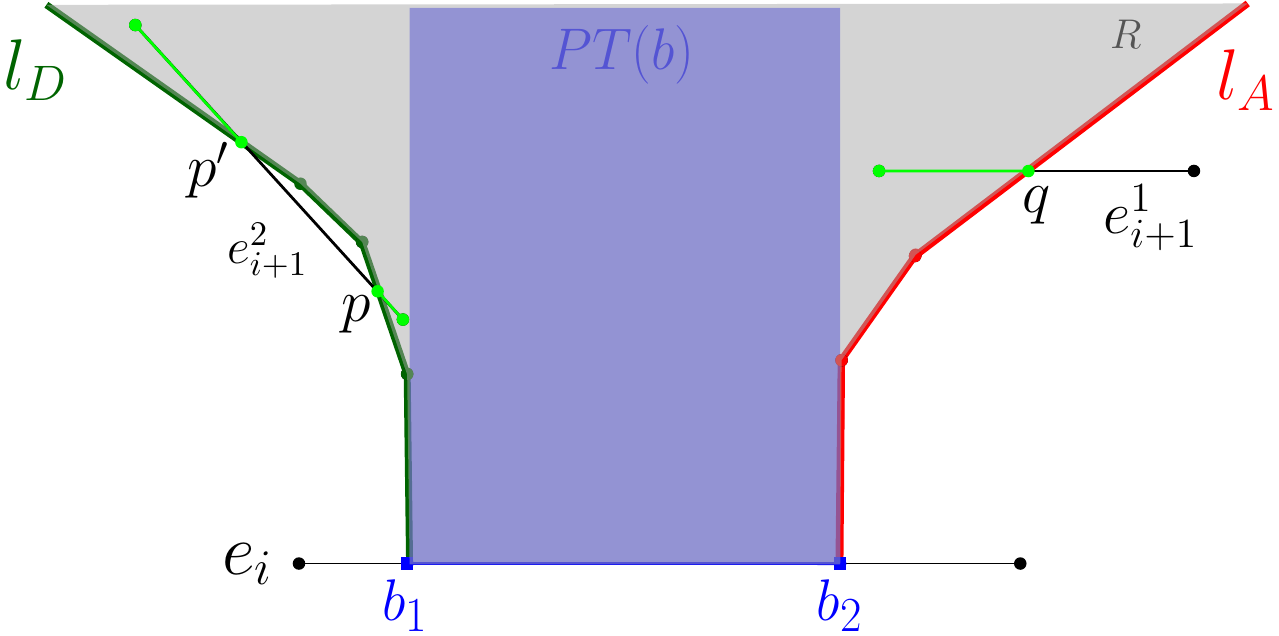}
  \newline
\newline
\newline
\newline
  \caption{Case~A.}
  \label{fig:e_intersect_R_a}
\end{subfigure}%
\hspace{3em}%
      \begin{subfigure}{.45\textwidth}
  \includegraphics[scale=0.35]{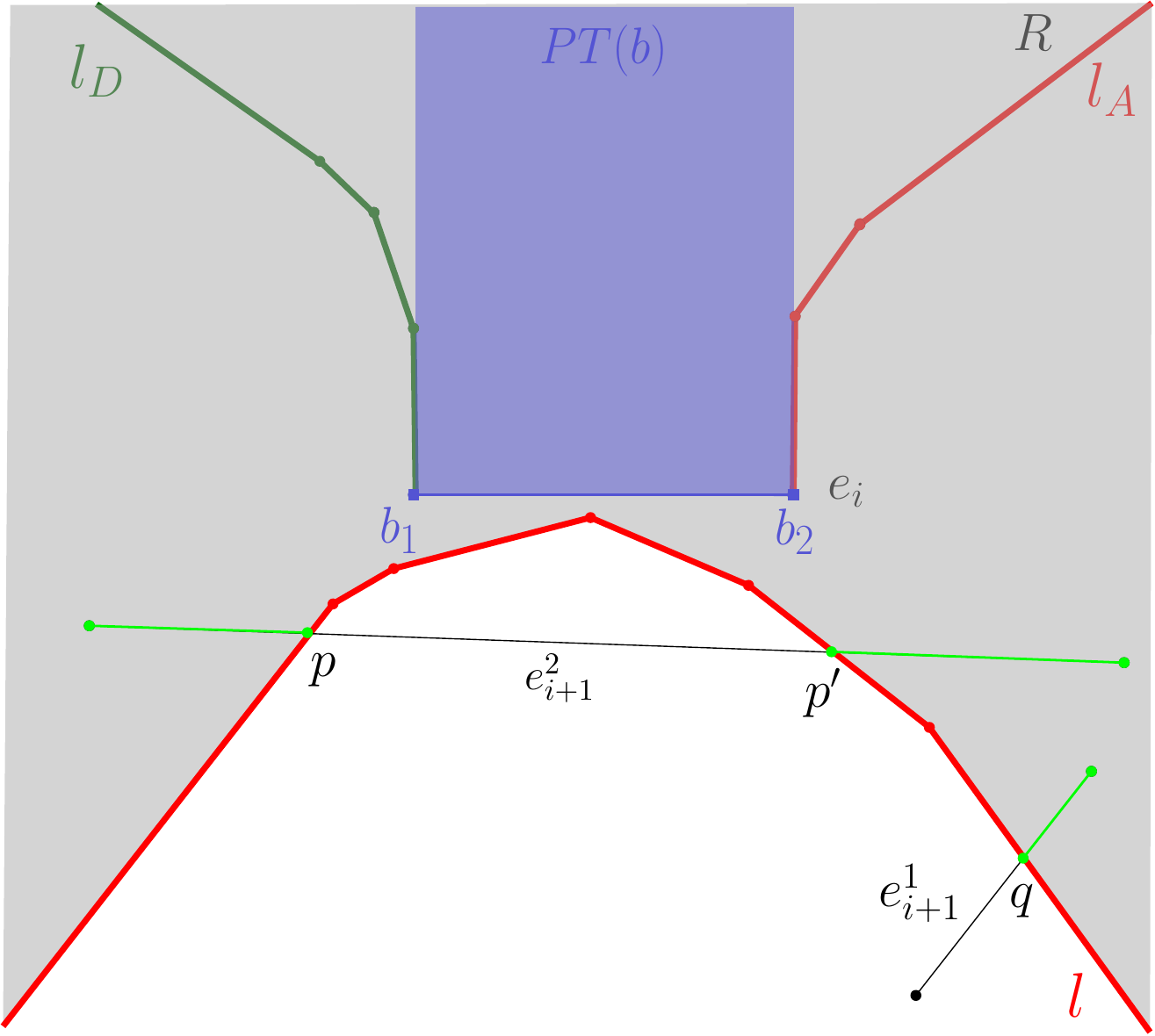}
  \caption{Case~D.}
  \label{fig:e_intersect_R_b}
  \end{subfigure}%
\caption{$e_{i+1}$ intersects $R$'s boundary either at a single point $q$ or at two points $p$ and $p'$.}
\label{fig:e_intersect_R}
\end{figure}

We have shown that by passing through an interval in $\{b\} \cup \delta(b) \setminus I(e_i,\overline{c_j})$, at most two intervals (with associated length $l(e_i)+1$ and orientation $c_j$) are formed on $e_{i+1}$. Moreover, each of these intervals contains an endpoint of $e_{i+1}$. Therefore, the total number of such intervals that are formed on $e_{i+1}$, by passing through an interval in $\bigcup_{b\in I^+(e_i,c_j)} \{b\} \cup \delta(b) \setminus I(e_i,\overline{c_j})$ is at most two. (For each endpoint $p$ of $e_{i+1}$, we retain only the longest interval with $p$ as one of its endpoints.) 

Finally, observe that by passing through an interval in $I(e_i,\overline{c_j})$ and turning backwards in orientation $c_j$, at most one interval is formed on $e_{i+1}$, which does not necessarily contain an endpoint of $e_{i+1}$.

We conclude that at most $|I(e_i,\overline{c_j})| + 2$ intervals (with associated length $l(e_i)+1$ and orientation $c_j$) are formed on $e_{i+1}$ during the execution of the algorithm (in the case that there is an interval $h \in I(e_i,\overline{c_j}$) such that $PT(h)\cap e_{i+1}\neq \emptyset$). We have used the equality $\bigcup_{b\in I^+(e_i,c_j)} \{b\} \cup \delta(b) = I^+(e_i,c_j) \cup I(e_i) \setminus I(e_i,c_j)$, which follows from Lemma~\ref{claim:lemma3.1}. 

We now proceed to the complementary case. 

\vspace{2mm}
\noindent
{\bf For any interval $h\in I(e_i,\overline{c_j}$), $PT(h)\cap e_{i+1} = \emptyset$.}
We defer the details of this case (which are similar to those of the previous case) to Appendix~\ref{app:second_case}.
These details lead to the conclusion that at most $|I^+(e_i,c_j)|+2$ intervals (with associated length $l(e_i)+1$ and orientation $c_j$) are formed on $e_{i+1}$ during the execution of the algorithm in this case.

\begin{figure}[h]
\begin{subfigure}{.45\textwidth}
\centering
  \includegraphics[scale=0.48]{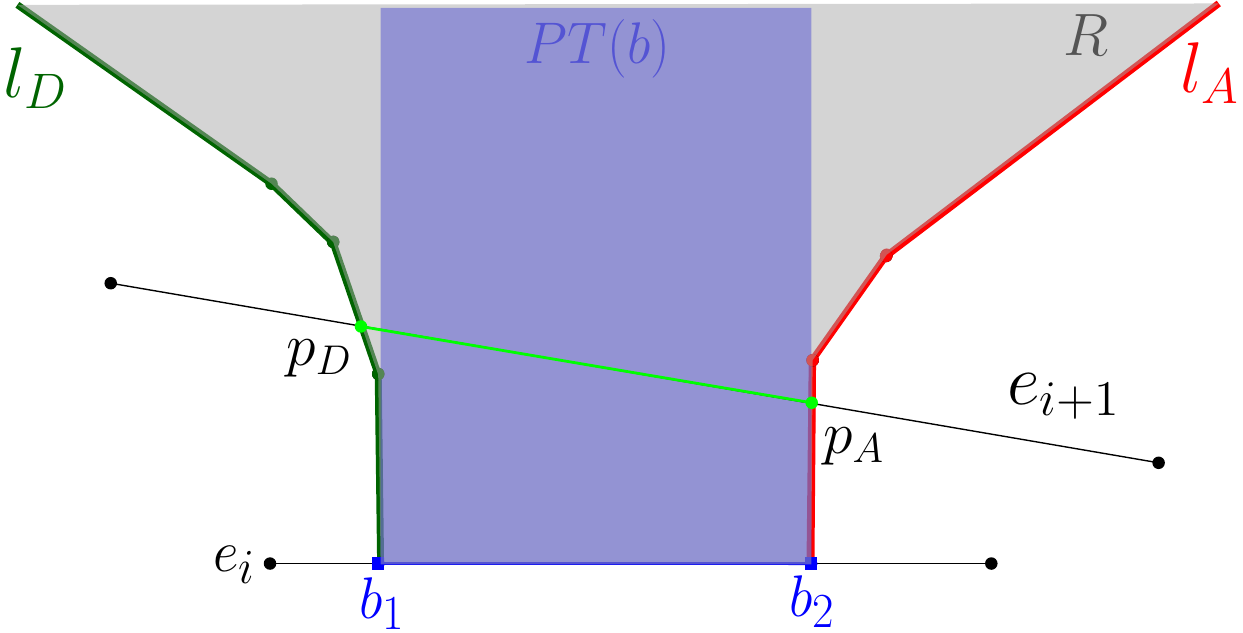}
  \newline
  \newline
  \newline
    \newline
  \caption{$e_{i+1}$ intersects $R$'s boundary at two points $p_A$ and $p_D$, creating a single internal interval (green) on $e_{i+1}$.}
  \label{fig:intersection_R_a}
\end{subfigure}%
\hspace{3em}%
      \begin{subfigure}{.45\textwidth}
  \includegraphics[scale=0.4]{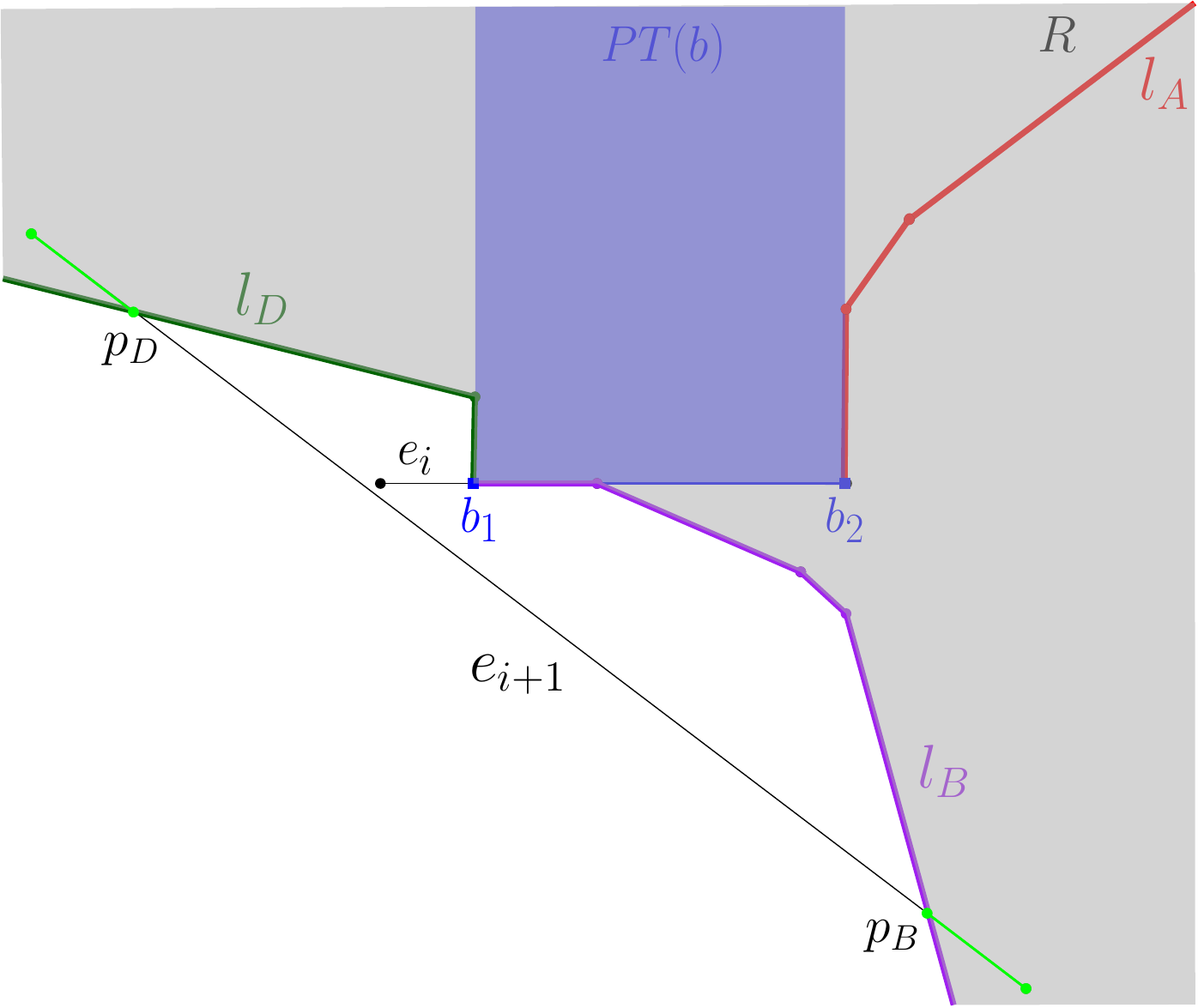}
  \caption{$e_{i+1}$ intersects $R$'s boundary at two points $p_B$ and $p_D$, creating two intervals (green), where each of them contains an endpoint of $e_{i+1}$.}
  \label{fig:intersection_R_b}
  \end{subfigure}%
\caption{Cases A and B, where there is no such interval $h$.}
\label{fig:intersection_R}
\end{figure}

Since only one of the two cases holds (i.e., either there is such an interval $h$ or there is not), we conclude that at most $max\left\{|I(e_i,\overline{c_j})|+2, |I^+(e_i,c_j)|+2 \right\}$ $= max\left\{ |I(e_i,\overline{c_j})|, |I^+(e_i,c_j)| \right\} +2$ intervals with associated length $l(e_i)+1$ and orientation $c_j$ are formed on $e_{i+1}$ during the execution of the algorithm. This completes the proof of Claim~\ref{claim:claim3}.
\end{proof}

\begin{lemma}
\label{claim:lemma4.1}
For any interval $a\in I(e_i)$ and orientation $c_j\in C$, we do not need to compute the interval on $e_{i+1}$ with associated length and orientation $l(e_i)+2$ and $c_j$, respectively, which is formed by passing through $a$ and then making two turns, where the first is in orientation $\overline{c_a}$.
\end{lemma}

\begin{proof}
By Claim~\ref{claim:claim2}, $l(e_i) \le l(e_{i+1}) \le l(e_i)+2$. So, the intervals on $e_{i+1}$ of length $l(e_i)+2$ are only relevant if  $l(e_{i+1}) > l(e_i)$ (Claim~\ref{claim:claim1}). Assume therefore that $l(e_{i+1}) > l(e_i)$, and let $a \in I(e_i)$ (e.g., the red interval in Figure~\ref{fig:pi_to_pi'}).
Let $\pi$ be a tour of $E$ that passes through $a$ at a point $p_i$, makes a turn in orientation $\overline{c_a}$ at point $p$, and makes another turn in orientation $c_j$ at point $p'$, such that $\pi_{i+1}$ (the portion of $\pi$ from $s$ to $e_{i+1}$) corresponds to an interval of length $l(e_i)+2$.

We distinguish between two cases.
If $pp' \cap e_i = \emptyset$ (i.e., the second turn is before $\pi$ crosses $e_i$ again), as shown in Figure~\ref{fig:path_pi}, then $pp'$ does not intersect $e_{i+1}$,
since this would imply $l(e_i)=l(e_{i+1})$. 
Therefore, $\pi$ reaches $e_{i+1}$ only after the turn at $p'$, and the tour $\pi'$ which is obtained from $\pi$ by deleting the link $pp'$ (see Figure~\ref{fig:path_pi'}), is a tour of $E$ of length $|\pi|-1$, hence $\pi$ is not a minimum-link tour of $E$.
Since our goal is to find a minimum-link tour of $E$, we do not need to compute the interval on $e_{i+1}$ formed by paths such as $\pi$ satisfying the condition above.
 \begin{figure}[!h]
 \begin{subfigure}{.5\textwidth}
  \includegraphics[scale=0.5]{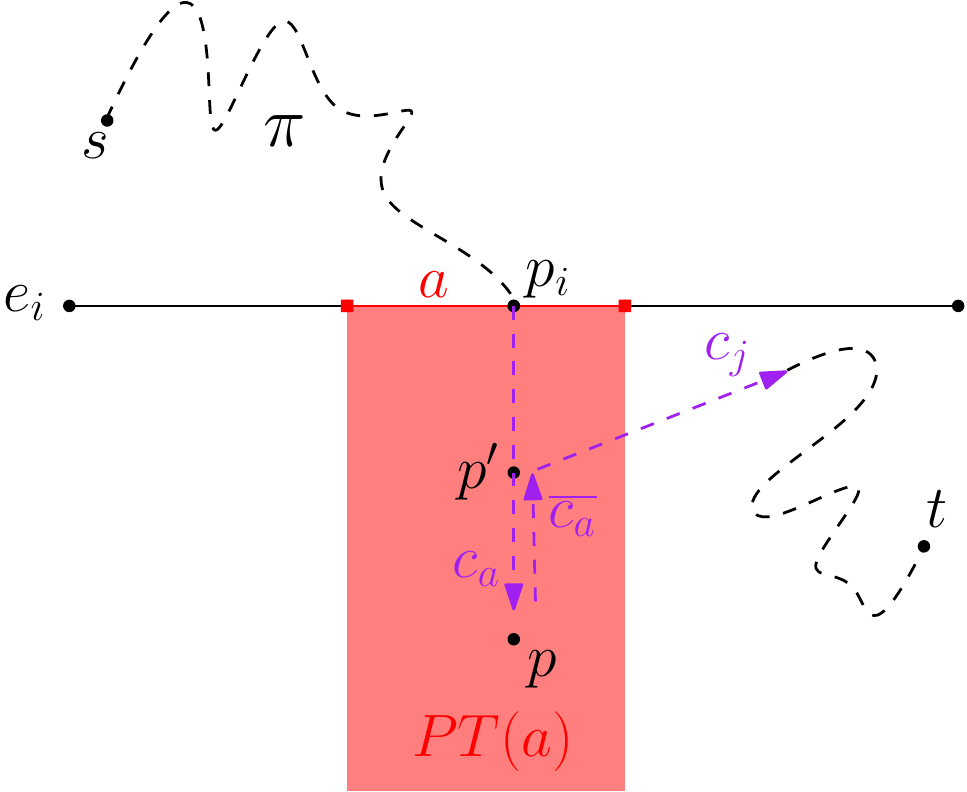}
  \caption{A tour $\pi$ that passes through $a\in I(e_i)$ (red), makes a turn in orientation $\overline{c_a}$, and another turn in orientation $c_j$.}
  \label{fig:path_pi}
\end{subfigure}%
 \hspace{2em}%
      \begin{subfigure}{.5\textwidth}        
       \includegraphics[scale=0.5]{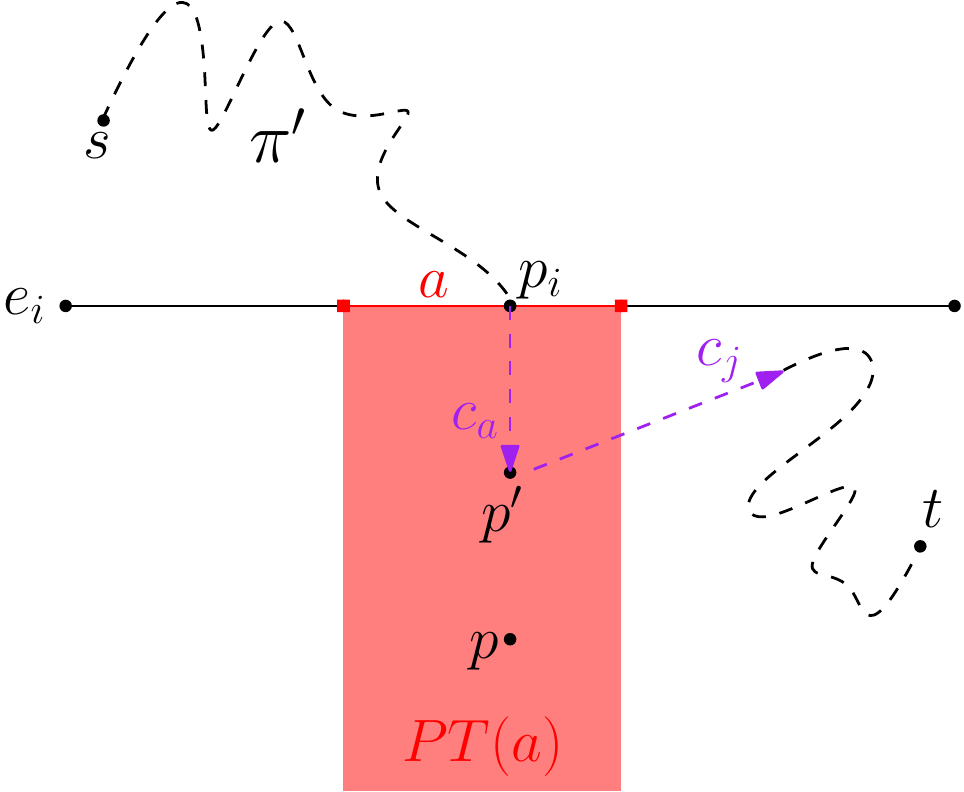}
        \caption{The tour $\pi'$ of length $|\pi|-1$, which is obtained by deleting the link $pp'$ from $\pi$.\newline}
  \label{fig:path_pi'}
\end{subfigure}%
\caption{Proof of Lemma~\ref{claim:lemma4.1}. The case where the second turn is before $\pi$ crosses $e_i$ again.}
\label{fig:pi_to_pi'}
\end{figure}

If $pp' \cap e_i \neq \emptyset$ (i.e., the second turn is not before $\pi$ crosses $e_i$ again), let $T$ denote the region of all points that can be reached by such paths, i.e., paths such as $\pi$ satisfying the condition above (see the orange region in Figure~\ref{fig:T_region_a}). Then $T \cap e_{i+1}$ is the interval on $e_{i+1}$ with associated length $l(e_i)+2$ and orientation $c_j$, formed by these paths.
But, by Lemma~\ref{claim:lemma3.1}, there exists $b\in I^+(e_i,\overline{c_a})$ such that $a\subseteq b$ (see the blue interval in Figure~\ref{fig:T_region_b}), and clearly $T \subseteq \psi(b,c_j)$, implying $T \cap e_{i+1} \subseteq \psi(b,c_j) \cap e_{i+1}$. The latter interval, i.e., $\psi(b,c_j) \cap e_{i+1}$  is computed by our algorithm, so we do not need to compute $T \cap e_{i+1}$.  

\begin{figure}[!h]
\begin{subfigure}{.5\textwidth}
  \includegraphics[scale=0.55]{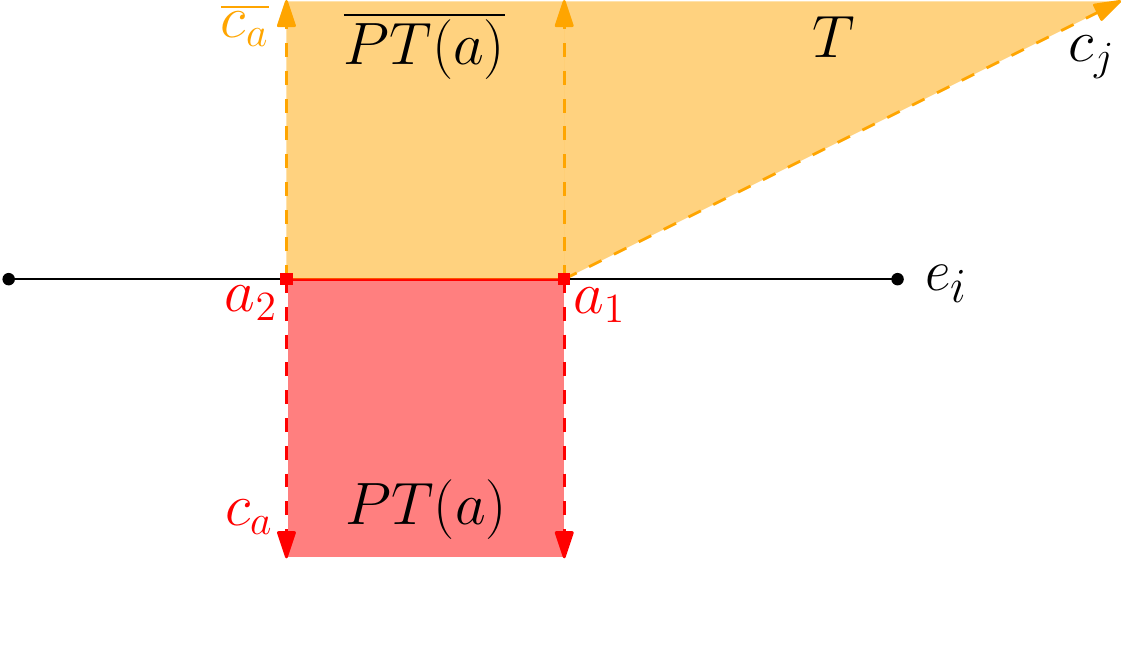}
  \caption{An interval $a$ (red) and the corresponding region $T$ (orange).}
  \label{fig:T_region_a}
\end{subfigure}%
 \hspace{2em}%
      \begin{subfigure}{.5\textwidth}
  \includegraphics[scale=0.55]{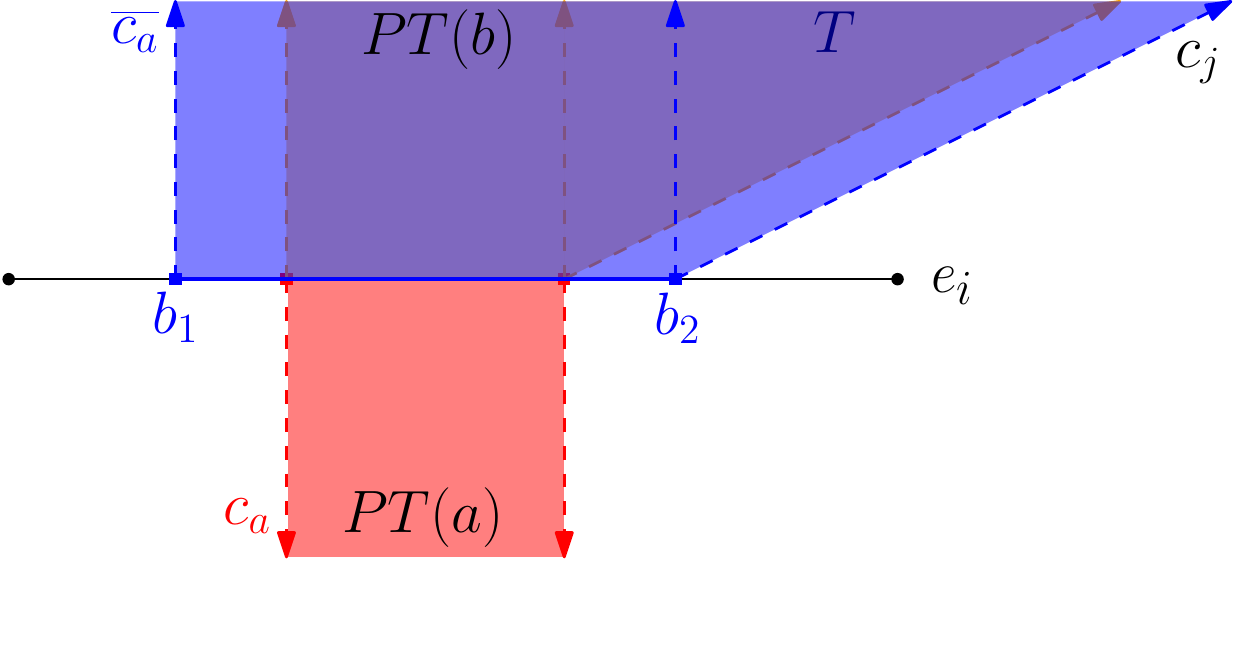}
  \caption{An interval $b$ and the corresponding region $\psi(b,c_j)$ (blue), which contains $T$.}
  \label{fig:T_region_b}
\end{subfigure}%
\caption{Proof of Lemma~\ref{claim:lemma4.1}. The case where the second turn is not before $\pi$ crosses $e_i$ again.}
\label{fig:T_region}
\vspace{-5mm}
\end{figure}
\end{proof}

\begin{lemma}
\label{claim:lemma4.2}
For any interval $a\in I(e_i)$, any point $p\in \mathbb{R}^2$, and any orientation $c_j \notin \{c_a, c_{a+1}, c_{a-1}\}$, $p$ can be reached by a path that passes through $a$ and then makes a turn in some orientation $c\neq \overline{c_a}$ and another turn in orientation $c_j$.
\end{lemma}
\begin{proof}
Consider any interval $a\in I(e_i)$. Recall that $\psi(a,\overline{c_{a}}_{+1})$ ($\psi(a,\overline{c_{a}}_{-1})$) denotes the region of all the points that can be reached by a path that passes through $a$ and then makes a turn in orientation $\overline{c_a}_{+1}$ ($\overline{c_a}_{-1}$) (see Figure~\ref{fig:delta_region}).
It is easy to see that $\psi(a,c)\subseteq \psi(a,\overline{c_{a}}_{+1})$ for any $c\in \phi(\overline{c_{a}},c_a)$ and $\psi(a,c)\subseteq \psi(a,\overline{c_{a}}_{-1})$ for any $c\in \phi(c_a, \overline{c_{a}})$.
Therefore, $\Delta_a= \psi(a,\overline{c_{a}}_{+1}) \cup  \psi(a,\overline{c_{a}}_{-1})$ is the region of all the points that can be reached by a path that passes through $a$ and then makes a turn in some orientation $c\neq \overline{c_a}$ (see Figure~\ref{fig:delta_region_c}).
\begin{figure}[ht]
\centering
\hspace{-3em}%
\begin{subfigure}{.27\textwidth}
  \includegraphics[scale=0.45]{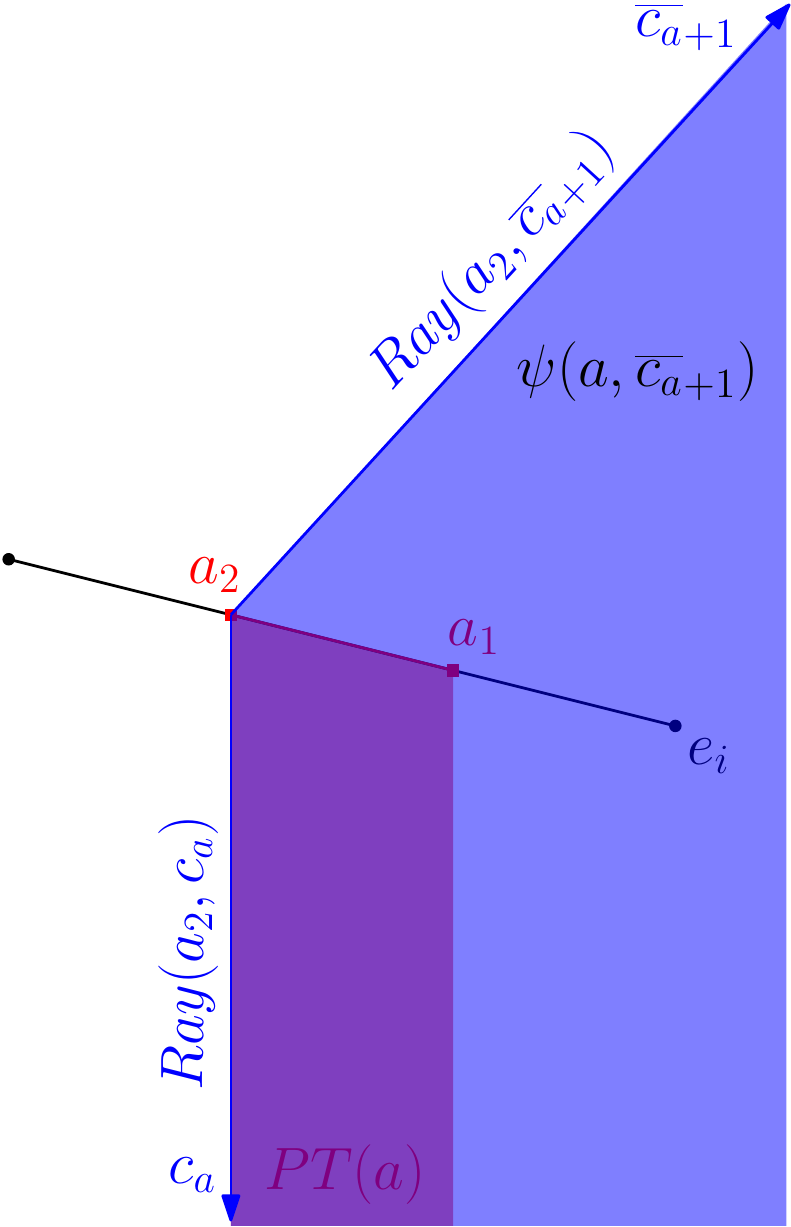}
  \caption{An interval $a$ (red) and the region $\psi(a,\overline{c_{a}}_{+1})$ (blue).}
  \label{fig:delta_region_a}
\end{subfigure}%
 \hspace{3em}%
      \begin{subfigure}{.27\textwidth}
  \includegraphics[scale=0.45]{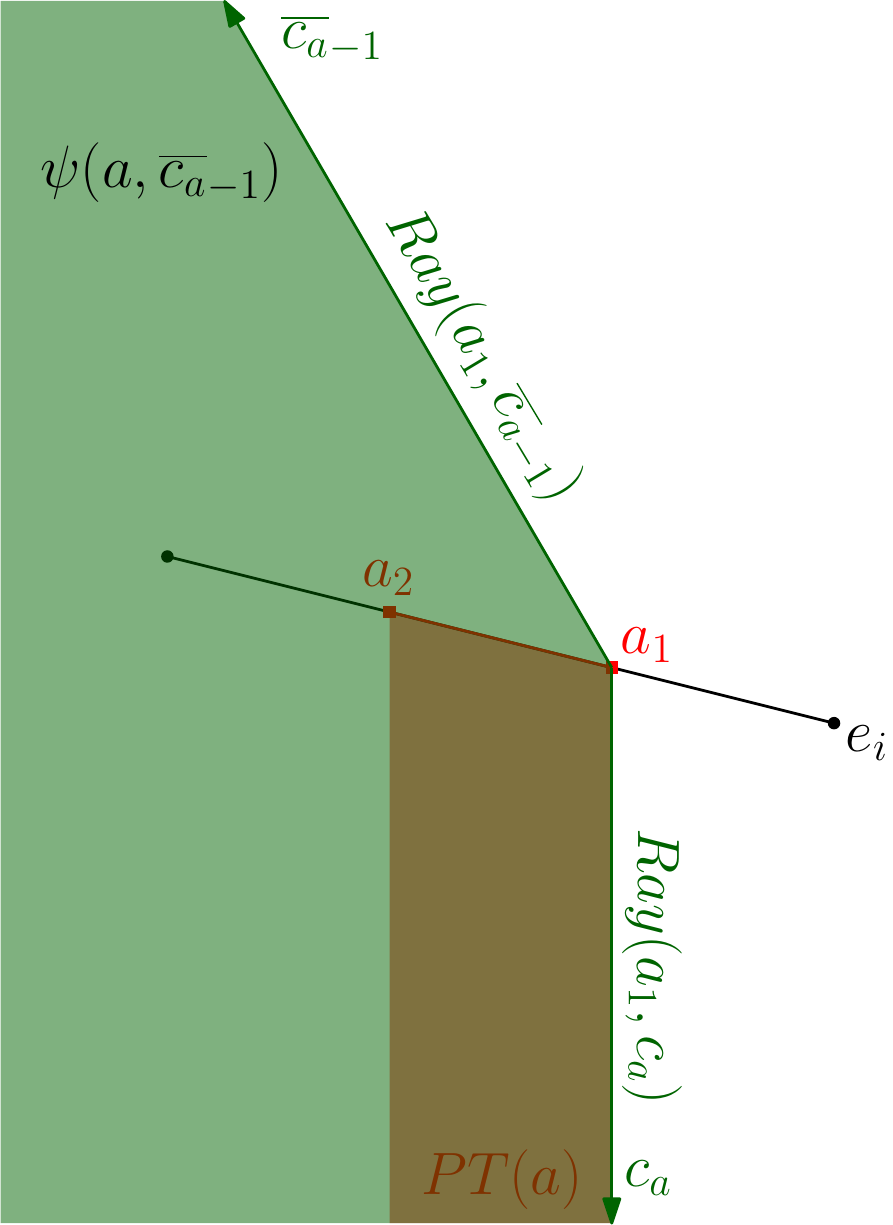}
  \caption{An interval $a$ (red) and the region $\psi(a,\overline{c_{a}}_{-1})$ (green).}
  \label{fig:delta_region_b}
\end{subfigure}%
\hspace{4em}%
\begin{subfigure}{0.35\textwidth}
  \includegraphics[scale=0.45]{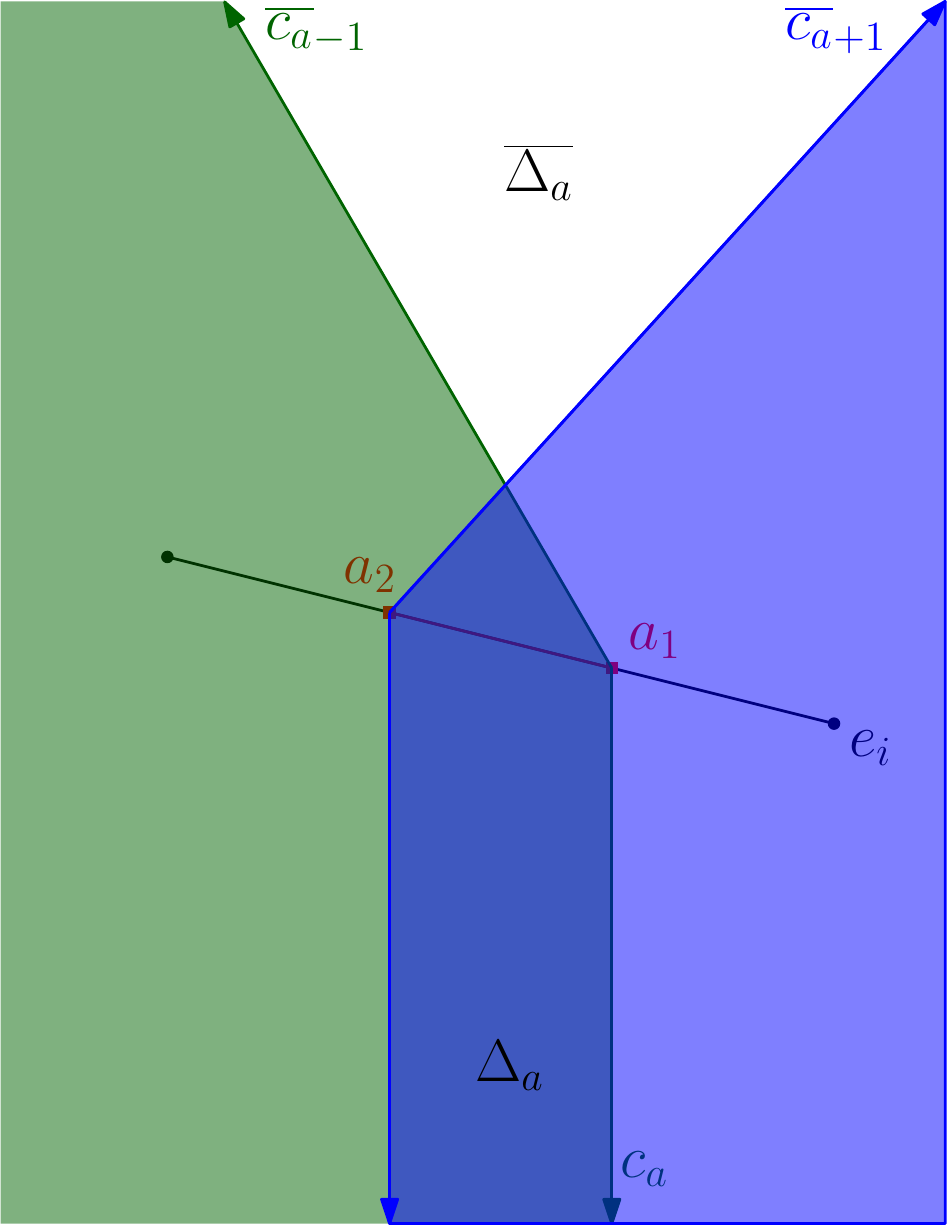}
  \caption{The region $\Delta_a$, which is the union of $\psi(a,\overline{c_{a}}_{+1})$ and $\psi(a,\overline{c_{a}}_{-1})$.}
\label{fig:delta_region_c}
\end{subfigure}%
\caption{All the points that can be reached by a path that passes through $a$ and then makes a turn in some orientation $c\neq \overline{c_a}$.}
\label{fig:delta_region}
\vspace{-5mm}
\end{figure}

Consider any point $p\in \mathbb{R}^2$ and any orientation $c_j \notin \{c_a,c_{a+1},c_{a-1}\}$. If $p\in \Delta_a$, then $p$ can be reached by a path that passes through $a$ and then makes a turn in some orientation $c\neq \overline{c_a}$. By making an additional turn at $p$ in orientation $c_j$ (without extending the path), we obtain a path that reaches $p$ as required.

If $p\in \overline{\Delta_a}= \mathbb{R}^2 \backslash \Delta_a$, then $Ray(p,\overline{c_j}) \cap \Delta_a \neq \emptyset$, since $\overline{c_j} \notin \{\overline{c_{a}}_{-1}, \overline{c_a}, \overline{c_{a}}_{+1}\}$ (as shown in Figure~\ref{fig:fig16}).
Let $p'$ be any point on $Ray(p,\overline{c_j}) \cap \Delta_a$, then $p'$ can be reached by a path that passes through $a$ and then makes a turn in some orientation $c\neq \overline{c_a}$, and by extending this path by adding the link $p'p$, we obtain a path that reaches $p$ as required.
\end{proof}

Consider the region $\Delta_a=\psi(a,\overline{c_{a}}_{+1}) \cup  \psi(a,\overline{c_{a}}_{-1})$ defined in the proof of Lemma~\ref{claim:lemma4.2}. Then, as mentioned in the proof of Lemma~\ref{claim:lemma4.2}, $\Delta_a$ is the region of all the points that can be reached by a path that passes through $a$ and then makes a turn in some orientation $c\neq \overline{c_a}$. In addition, we notice that by extending such a path by adding a link in orientation $c_j$, for $c_j \in \{c_a,c_{a+1},c_{a-1}\}$, we cannot leave $\Delta_a$ (see Figure~\ref{fig:fig17}), since for any point $q \in \Delta_a$, $Ray(q,c_j) \subseteq \Delta_a$.

\begin{figure}[h]
    \begin{minipage}{.4\textwidth}
    \centering
    \includegraphics[scale=0.40]{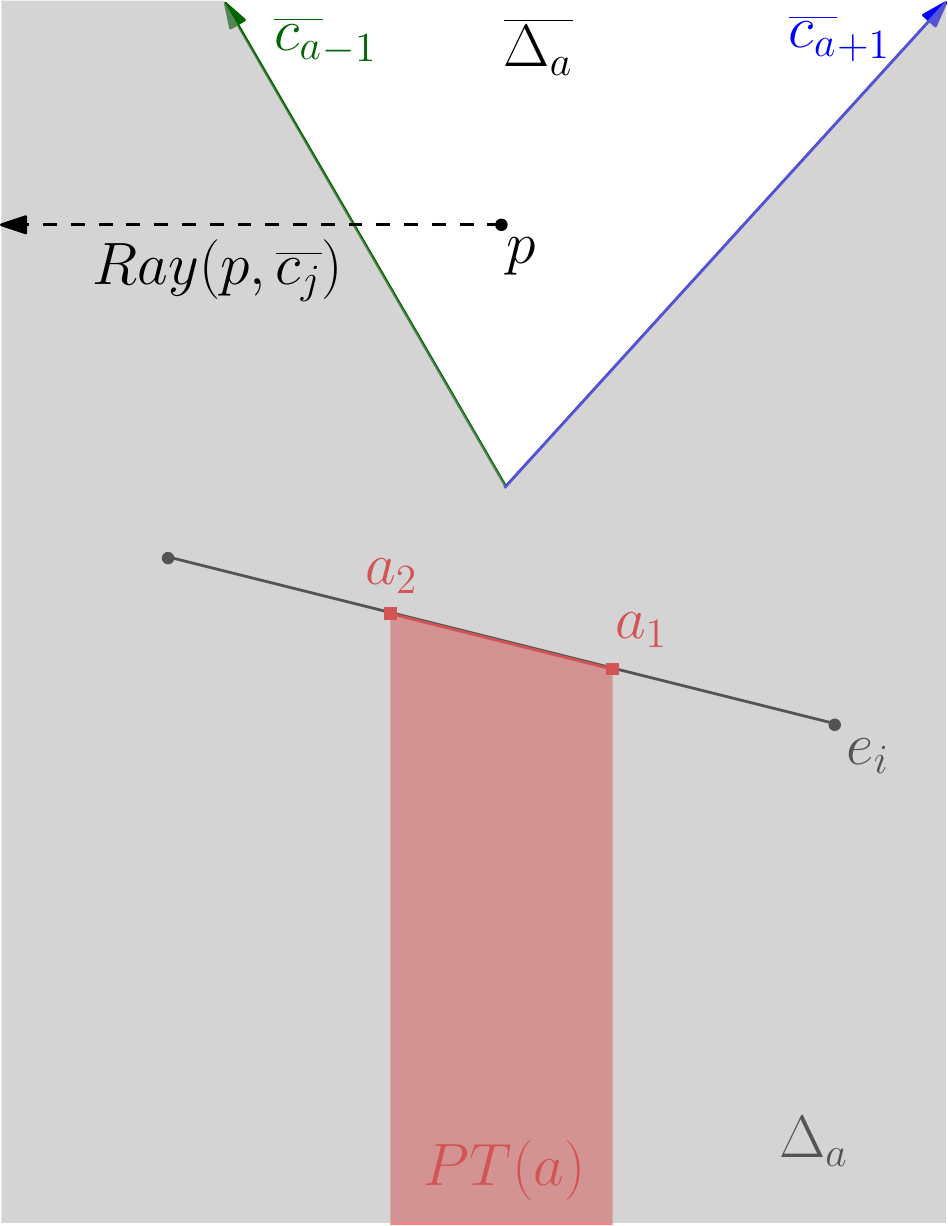}
    \caption{Lemma~\ref{claim:lemma4.2}. If $p\in \overline{\Delta_a}$, then $Ray(p,\overline{c_j}) \cap \Delta_a \neq \emptyset$.}
    \label{fig:fig16}
    \end{minipage}%
     \hspace{5em}%
    \begin{minipage}{.51\textwidth}
    \centering
    \includegraphics[scale=0.40]{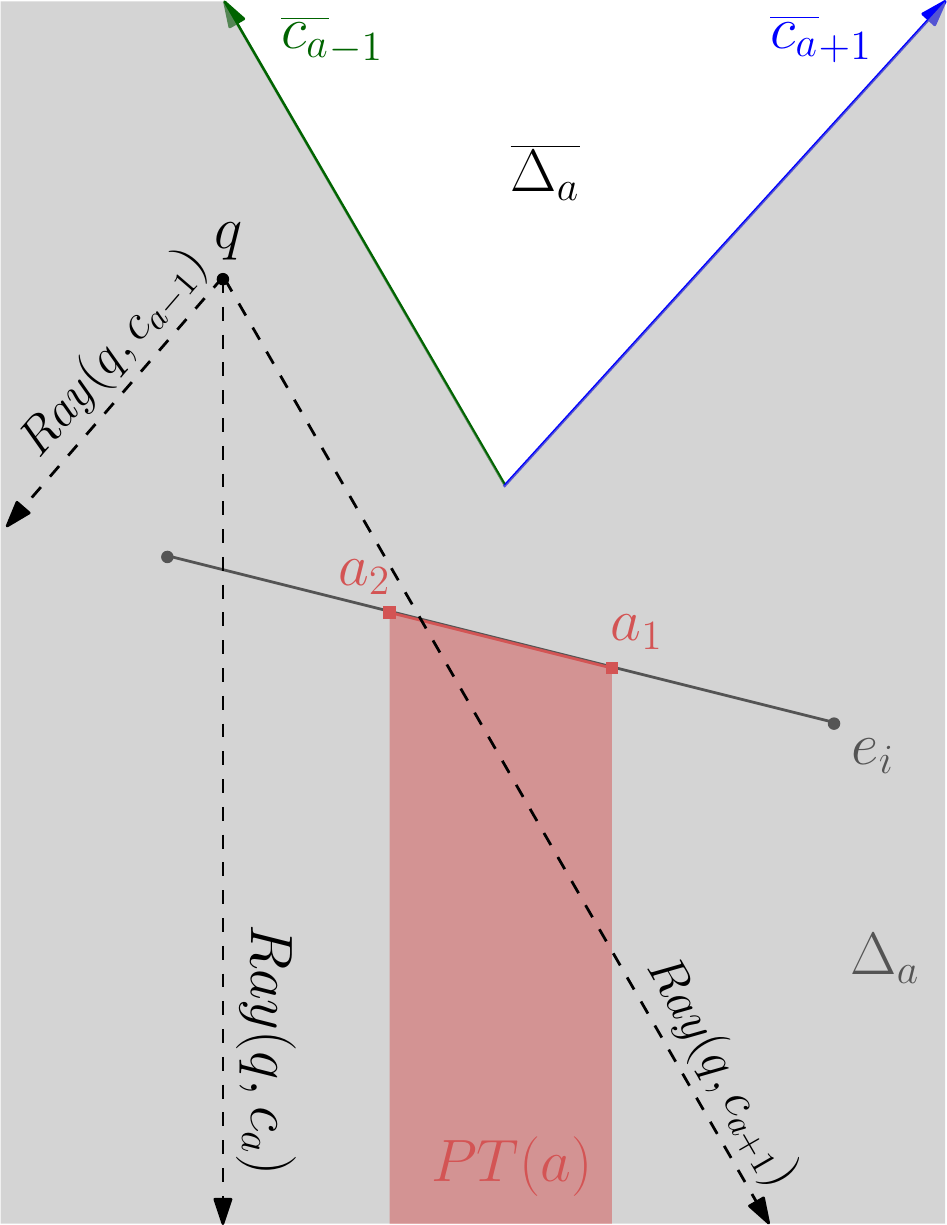}
    \caption{For any $q\in \Delta_a$, $Ray(q,c_j)\cap \overline{\Delta_a} = \emptyset$, for $c_j\in \{c_a, c_{a+1}, c_{a-1}\}$.}
    \label{fig:fig17}
    \end{minipage}%
\vspace{-5mm}
\end{figure}

The following claim bounds the number of intervals with associated length and orientation $l(e_i)+2$ and $c_j$, respectively, that are `created' on $e_{i+1}$.
\begin{my_claim}
\label{claim:claim4}
At most $|I^+(e_i,\overline{c_j})| + 2$ intervals with associated length and orientation $l(e_i)+2$ and $c_j$, respectively, are `created' on $e_{i+1}$, during the execution of the algorithm.
\end{my_claim}
\begin{proof}
The proof can be found in Appendix~\ref{app:claim4}. Here, we only observe that
there are two ways to reach a point on $e_{i+1}$ with a path of length $l(e_i)+2$ whose last link is of orientation $c_j$. The first is by passing through one of the intervals $a\in I(e_i)$ and then making two turns, where the first one is in orientation $c\neq \overline{c_a}$ and the second one is in orientation $c_j$ (see Lemma~\ref{claim:lemma4.1}).
The second way is by passing through one of the intervals in $I^+(e_i)\setminus I^+(e_i,c_j)$, and then making a turn in orientation $c_j$.
That is, the intervals on $e_{i+1}$ with associated length $l(e_i)+2$ and associated orientation $c_j$ are determined by the intervals in $I(e_i) \cup (I^+(e_i)\setminus I^+(e_i,c_j))$.
\end{proof}

The following claim bounds the number of intervals with associated length and orientation $l(e_i)+3$ and $c_j$, respectively, that are `created' on $e_{i+1}$.

\begin{my_claim}
\label{claim:claim5}
For any $q\in e_{i+1}$ and for any $c_j\in C$, there exists a path of length $l(e_i)+3$ from $s$ to $q$, whose last link has orientation $c_j$, for $1\leq i \leq n-1$. 
\end{my_claim}
\begin{proof}
Consider any path $\pi$ from $s$ to $e_i$ that corresponds to an interval in $I(e_i)$, and let $p$ be the point on $e_i$ where $\pi$ ends. Since $C$ spans the plane, there exists a two-link path from $p$ to $q$, and
by making a turn at $q$ in orientation $c_j$ (without extending the path), we obtain a three-link path $\pi_{p,q}$ from $p$ to $q$ whose last link has orientation $c_j$.
So, the path obtained by concatenating the paths $\pi$ and $\pi_{p,q}$ is as desired.
\end{proof}

\begin{my_claim}
\label{claim:claim6}
For any $0 \leq i \leq n+1$ and $c_j \in C$, $|I(e_i,c_j)| \le 2i+1$ and $|I^+(e_i,c_j)| \le 2i+1$.
\end{my_claim}
\begin{proof}
The proof is by induction on $i$.
For $i=0$, the claim is clearly true;  $|I(e_0,c_j)| = |I^+(e_0,c_j)| = 1$.

Assume now that the claim is true for $i$, $0 \le i \le n$, that is, for any $c_j \in C$, we have $|I(e_i,c_j)| \le 2i+1$ and $|I^+(e_i,c_j)| \le 2i+1$. We show below that it remains true for $i+1$.

Recall that $l(e_i) \le l(e_{i+1}) \le l(e_i) + 2$ (Claim~\ref{claim:claim2}). We show that the claim remains true in each of the resulting three cases.  
{\begin{itemize}
    \item {\bf Case A}: $l(e_{i+1})=l(e_i)$. In this case $I(e_{i+1},c_j)$ stores the +0-intervals on $e_{i+1}$ with respect to $I(e_i,c_j)$.
    Since, each interval $a \in I(e_i,c_j)$ `creates' at most one +0-interval on $e_{i+1}$, we get that $|I(e_{i+1},c_j)| \le |I(e_i,c_j)| \le 2i + 1$. 

    Recall that $I^+(e_{i+1},c_j)$ is the set of maximal intervals on $e_{i+1}$ formed by all paths of length $l(e_{i+1})+1 = l(e_i)+ 1$, whose last link has orientation $c_j$. By Claim~\ref{claim:claim3},  $|I^+(e_{i+1},c_j)| \le \max\{ |I(e_i,\overline{c_j})|,|I^+(e_i,c_j)|\}+2$, and therefore $|I^+(e_{i+1},c_j)| \le \max\{ 2i+1,2i+1\}+2 = 2(i+1)+1$.

    \item {\bf Case B:} $l(e_{i+1})=l(e_i)+1$. In this case, $I(e_{i+1},c_j)$
    is the set of maximal intervals on $e_{i+1}$ formed by all paths of length $l(e_{i+1}) = l(e_i)+ 1$, whose last link has orientation $c_j$. By Claim~\ref{claim:claim3},  $|I(e_{i+1},c_j)| \le \max\{ |I(e_i,\overline{c_j})|,|I^+(e_i,c_j)|\}+2$, so,  $|I(e_{i+1},c_j)| \le \max\{ 2i+1,2i+1\}+2 = 2(i+1)+1$. 
    
    Now, $I^+(e_{i+1},c_j)$ is the set of maximal intervals on $e_{i+1}$ formed by all paths of length $l(e_{i+1})+1 = l(e_i)+ 2$, whose last link has orientation $c_j$. By Claim~\ref{claim:claim4}, $|I^+(e_{i+1},c_j)| \le |I^+(e_i,\overline{c_j})| + 2$, and therefore $|I^+(e_{i+1},c_j)| \le 2i+1+2 = 2(i+1)+1$.
    
    \item {\bf Case C:} $l(e_{i+1})=l(e_i)+2$. In this case, $I(e_{i+1},c_j)$
    is the set of maximal intervals on $e_{i+1}$ formed by all paths of length $l(e_{i+1}) = l(e_i)+ 2$, whose last link has orientation $c_j$. Thus, by Claim~\ref{claim:claim4}, $|I(e_{i+1},c_j)| \le |I^+(e_i,\overline{c_j})| + 2 \le 2(i+1)+1$.
    Moreover, in this case, $I^+(e_{i+1},c_j) = \{e_{i+1}\}$, so $|I^+(e_{i+1},c_j)| = 1$. 
\end{itemize}}
\end{proof}

\noindent
{\bf Running time.}
We bound the running time of each of the two stages of our algorithm.
Consider the $i$'th iteration of the main loop of Stage~I. We need $O(|I(e_i,c_j)|+|I^+(e_i,c_j)|)$ time to compute the $+0$-intervals on $e_{i+1}$, $O(|I(e_i) \setminus I(e_i,c_j)|+|I^+(e_i) \setminus I^+(e_i,c_j)|)$ time to compute the $+1$-intervals, and $O(|I(e_i)|)$ time to compute the $+2$-intervals. 
Since we perform this calculation for each $c_j\in C$, the running time of the i'th iteration is $O(|C| \cdot \{|I(e_i)|+|I^+(e_i)|\})$. 
By Claim~\ref{claim:claim6} we conclude that $|I(e_i)| = O(|C| \cdot (2i+1))$ and $|I^+(e_i)| = O(|C| \cdot (2i+1))$, for $1\leq i \leq n+1$.
Therefore, the running time of Stage~I is $\sum_{i=1}^{n+1} O(|C| \cdot |C| \cdot (2i+1)) = O(|C|^2 \cdot n^2)$.

In stage 2, we run Algorithm~\ref{alg:Recovery} for each $i$ from $n + 1$ to $1$. The running time of Algorithm~\ref{alg:Recovery} is $O(|I(e_{i-1})|+|I^+(e_{i-1})|)$, and by Claim~\ref{claim:claim6} we get $O(|C| \cdot i)$. Therefore, the running time of Stage~II is $\sum_{i=1}^{n+1} O(|C| \cdot i) = O(|C| \cdot n^2)$.

Thus, the overall running time of the algorithm is $O(|C|^2 \cdot n^2)$, as summarized:

\begin{theorem}
Given a set $E$ of $n$ disjoint $C$-oriented segments in the plane and points $s$ and $t$ that do not belong to any of the segments in $E$, one can compute a minimum-link $C$-oriented tour of $E$ in $O(|C|^2 \cdot n^2)$ time.
\end{theorem}

\vspace{-3mm}
\section{Extensions}
In the case that $|C|=4$ (e.g., axis-parallel paths and segments), the specialization of our analysis shows a constant upper bound on the number of intervals on each segment; this results in overall time $O(n)$. Also, our analysis only required that consecutive segments in $E$ do not intersect each other; they can otherwise intersect.  In ongoing and future work we consider more general polygonal regions, possibly overlapping arbitrarily. We also consider query versions of the problem in which we build data structures (shortest path maps) that allow link distance queries on subsequences of the input set of regions, between query points in the plane.  Future work might examine problems in 3D.

\subsubsection*{Acknowledgements}
M. Katz was partially supported by the US-Israel Binational Science Foundation (BSF project 2019715 / NSF CCF-2008551).
J. Mitchell was partially supported by the National Science Foundation (CCF-2007275) and the US-Israel Binational Science Foundation (BSF project 2016116).

\bibliographystyle{plain}

\newpage
\appendix

\section{Details}
\subsection{Stage II}
\label{app:stage2}

In this stage we use the information collected in the first stage to construct a minimum-link tour $\pi$ of $E$.

We construct $\pi$ incrementally beginning at $t$ and ending at $s$. That is, in the first iteration we add the portion of $\pi$ from $t$ to $e_n$, in the second iteration we add the portion from $e_n$ to $e_{n-1}$, etc. 
Assume that we have already constructed the portion of $\pi$ from $t$ to $e_i$, where this portion ends at point $p$ of interval $a$ on $e_i$. We describe in Algorithm~\ref{alg:Recovery} how to compute the portion from $e_i$ to $e_{i-1}$, which begins at point $p$ of $a$ and ends at a point $p'$ of interval $b$ on $e_{i-1}$ (where $b \in I(e_{i-1}) \cup I^+(e_{i-1})$) and consists of $l_a-l_b+1$ links. 
Before continuing to the next iteration, we set $p=p'$ and $a=b$.

After adding the last portion, which ends at $s$, we remove all the redundant vertices from $\pi$, i.e., vertices at which $\pi$ does not make a turn.

\begin{algorithm}
\SetKwInOut{Input}{Input}
\SetKwInOut{Output}{Output}
\Input{Index $i$, an interval $a$ on $e_i$, a point $p \in a$}
\Output{An interval $b$ on $e_{i-1}$, a point $p'\in b$ and a $C$-oriented path $\pi_{p,p'}$ from $p$ to $p'$ consisting of $l_a-l_b+1$ links}
\BlankLine
\Begin{
\If {$l_a = l(e_{i-1})$}{
    $p' \gets Ray(p,\overline{c_a})\cap e_{i-1}$ and $\pi_{p,p'}\gets (p,p')$\;
    $b\gets$ the interval of $I(e_{i-1},c_a)$ containing $p'$\;
    \textbf{Stop and return} $p', b$ and $\pi_{p,p'}$\;
}
\If{$l_a = l(e_{i-1})+1$}{
    \If{$Ray(p,\overline{c_a})\cap e_{i-1}\neq \emptyset$}{
       $p'\gets Ray(p,\overline{c_a})\cap e_{i-1}$ and $\pi_{p,p'}\gets (p,p')$\;
       $b\gets$ the interval of $I^+(e_{i-1},c_a)$ containing $p'$\;
       \textbf{Stop and return} $p', b$ and $\pi_{p,p'}$\;
    }
    \Else{
       \ForEach{$d\in I(e_{i-1})\backslash I(e_{i-1},c_a)$}{
          \If{$Ray(p,\overline{c_a})\cap PT(d)\neq \emptyset$}{
             $q\gets$ any point in $Ray(p,\overline{c_a})\cap PT(d)$\;
			 $p'\gets Ray(q,\overline{c_d})\cap e_{i-1}$\;
			 $\pi_{p,p'}\gets (p,q,p')$ and $b\gets d$\;
			 \textbf{Stop and return} $p', b$ and $\pi_{p,p'}$\;
		   }
		}
	}
}
\If{$l_a$ = $l(e_{i-1})+3$}{
   $a \gets$ any interval in $I(e_i)$ that contains $p$\;
}
\If{$l_a = l(e_{i-1})+2$}{
   \ForEach{$d\in I^+(e_{i-1})\backslash I^+(e_{i-1},c_a)$}{
      \If{$Ray(p,\overline{c_a})\cap PT(d)\neq \emptyset$}{
         $q \gets$ any point in $Ray(p,\overline{c_a})\cap PT(d)$\;
         $p'\gets Ray(q,\overline{c_d})\cap e_{i-1}$\;
         $\pi_{p,p'}\gets (p,q,p')$ and $b\gets d$\;
         \textbf{Stop and return} $p', b$ and $\pi_{p,p'}$\;
      }
    }
    \ForEach{$d\in I(e_{i-1})$}{
       \If{$Ray(p,\overline{c_a})\cap (\psi(d,\overline{c_{d}}_{+1}) \cup \psi(d,\overline{c_{d}}_{-1})) \neq \emptyset$}{
          $v\gets$ any point in $Ray(p,\overline{c_a})\cap (\psi(d,\overline{c_{d}}_{+1}) \cup \psi(d,\overline{c_{d}}_{-1}))$\;
          \If{$Ray(v,c_{d+1})\cap PT(d)\neq \emptyset$}{
             $q\gets$ any point in $Ray(v,c_{d+1})\cap PT(d)$\;
             $p'\gets Ray(q,\overline{c_d})\cap e_{i-1}$, $b\gets d$\;
          }
          \Else{
 			$q\gets$ any point in $Ray(v,c_{d-1})\cap PT(d)$\;
			$p'\gets Ray(q,\overline{c_d})\cap e_{i-1}$, $b\gets d$\;
		  }
		  $\pi_{p,p'}\gets (p,v,q,p')$\;
		  \textbf{Stop and return} $p', b$ and $\pi_{p,p'}$\;
	    }
	}
 }
}
\caption{SubPath Recovery}
\label{alg:Recovery}
\end{algorithm}

\subsection{Proof of Claim~\ref{claim:claim3}: For any interval $h\in I(e_i,\overline{c_j}$), $PT(h)\cap e_{i+1} = \emptyset$}
\label{app:second_case}

{\bf If Case A:}
If $e_{i+1}$ intersects at most one of the chains $l_A$ and $l_D$, then, as in Case~A above, at most two intervals are formed on $e_{i+1}$, where each of them contains an endpoint of $e_{i+1}$ (see Figure~\ref{fig:e_intersect_R_a}).

If, however, $e_{i+1}$ intersects both $l_A$ and $l_D$ (which is possible since now $e_{i+1}$ may intersect $PT(b)$), then it intersects each of them exactly once (otherwise, we get that $e_{i+1} \cap b \neq \emptyset$), and a single interval is formed on $e_{i+1}$, which does not contain an endpoint of $e_{i+1}$ (see Figure~\ref{fig:intersection_R_a}).

{\bf If Case B:}
If $e_{i+1}$ intersects $l_B$ and $l_D$ in more than two points (in total), then it must also intersect $b$, which is impossible. Therefore,  
$e_{i+1}$ can intersect $l_B$ and $l_D$ in at most two points (in total), and at most two intervals are formed on $e_{i+1}$, where each of them contains an endpoint of $e_{i+1}$ (see Figure~\ref{fig:intersection_R_b}).

{\bf If Case C:}
Since Case~B and Case~C are symmetric, at most two intervals are formed on $e_{i+1}$, each of which contains an endpoint of $e_{i+1}$.

{\bf If Case D:} Identical to Case~D above, that is, $e_{i+1}$ intersects $l$ either at one or at two points, and at most two intervals are formed on $e_{i+1}$, each of which contains an endpoint of $e_{i+1}$ (see Figure~\ref{fig:e_intersect_R_b}).

We have shown that by passing through an interval in $\{b\} \cup \delta(b) \setminus I(e_i,\overline{c_j})$, at most one interval that does not contain an endpoint of $e_{i+1}$ is formed on $e_{i+1}$, or at most two intervals (with associated length $l(e_i)+1$ and orientation $c_j$) are formed on $e_{i+1}$, where each of them contains an endpoint of $e_{i+1}$.

Therefore, the total number of such intervals that are formed on $e_{i+1}$, by passing through an interval in $\bigcup_{b\in I^+(e_i,c_j)} \{b\} \cup \delta(b) \setminus I(e_i,\overline{c_j})$ is at most $|I^+(e_i,c_j)|+2$. (For each endpoint $p$ of $e_{i+1}$, we retain only the longest interval with $p$ as one of its endpoints.) 

Finally, observe that by passing through an interval $h \in I(e_i,\overline{c_j})$ and turning backwards in orientation $c_j$, at most one interval is formed on $e_{i+1}$. However, since $PT(h) \cap e_{i+1} = \emptyset$, the interval formed on $e_{i+1}$ is contained in $PT(b) \cap e_{i+1}$ and will therefore not appear in the returned set of intervals formed on $e_{i+1}$. 

We conclude that at most $|I^+(e_i,c_j)|+2$ intervals (with associated length $l(e_i)+1$ and orientation $c_j$) are formed on $e_{i+1}$ during the execution of the algorithm (in the case that for any interval $h\in I(e_i,\overline{c_j}$), $PT(h)\cap e_{i+1} = \emptyset$).

\subsection{Proof of Claim~\ref{claim:claim4}}
\label{app:claim4}
\begin{proof}
There are two ways to reach a point on $e_{i+1}$ with a path of length $l(e_i)+2$ whose last link is of orientation $c_j$. The first is by passing through one of the intervals $a\in I(e_i)$ and then making two turns, where the first one is in orientation $c\neq \overline{c_a}$ and the second one is in orientation $c_j$ (see Lemma~\ref{claim:lemma4.1}).
The second way is by passing through one of the intervals in $I^+(e_i)\setminus I^+(e_i,c_j)$, and then making a turn in orientation $c_j$.
That is, the intervals on $e_{i+1}$ with associated length $l(e_i)+2$ and associated orientation $c_j$ are determined by the intervals in $I(e_i) \cup (I^+(e_i)\setminus I^+(e_i,c_j))$. 

First, if there is an interval $a\in I(e_i)$ such that $c_j \notin \{c_{a-1}, c_a, c_{a+1}\}$, then, by Lemma~\ref{claim:lemma4.2}, $R_a=\mathbb{R}^2$, where $R_a$ is the region of all the points that can be reached by a path that passes through $a$ and then makes two turns, where the first is in orientation $c \neq \overline{c_a}$ and the second is in orientation $c_j$. Thus, $R_a \cap e_{i+1} = e_{i+1}$ and only one interval is formed on $e_{i+1}$.

Assume therefore that for each interval $a \in I(e_i)$, we have $c_j \in \{c_{a-1}, c_a, c_{a+1}\}$.
Now, fix an interval $a \in I(e_i)$, then $R_a = \Delta_a$ (as observed in the paragraph following Lemma~\ref{claim:lemma4.2}).
Consider an interval $b\in I^+(e_i) \setminus \{ I^+(e_i, \overline{c_j})\cup I^+(e_i,c_j) \}$ (we will treat the case $b \in I^+(e_i, \overline{c_j})$ separately), and set $R_{a\cup b}= R_a \cup \psi(b,c_j)$. Then $R_{a\cup b}\cap e_{i+1}$ gives us the intervals on $e_{i+1}$ with length $l(e_i)+2$ and orientation $c_j$, which are formed by passing through $a$ or $b$.

Next, we characterize the region $R_{a\cup b}$, by considering several cases.

{\bf Case A: $PT(b) \subseteq \Delta_a$} (see Figure~\ref{fig:delta_b_case_a}).\\
 Since for any $q \in \Delta_a$, $Ray(q,c_j) \subseteq \Delta_a$, we have  $\psi(b,c_j) = \bigcup_{q \in PT(b)} Ray(q,c_j) \subseteq \Delta_a$, and therefore $R_{a\cup b}=\Delta_a$.

{\bf Case B: $PT(b) \not\subseteq \Delta_a$ and $c_b \notin \{\overline{c_a}_{-1},\overline{c_a},\overline{c_a}_{+1}\}$} (see Figure~\ref{fig:delta_b_case_b}).\\
In this case, $R_{a\cup b}$ is the union  of $\Delta_a$ and the triangle that is formed by the rays defining $\Delta_a$ and the appropriate side of $PT(b)$; see the yellow triangle denoted $\Delta_b$ in Figure~\ref{fig:delta_b_case_b}.

\begin{figure}[!h]
\begin{subfigure}{.5\textwidth}
\centering
  \includegraphics[scale=0.55]{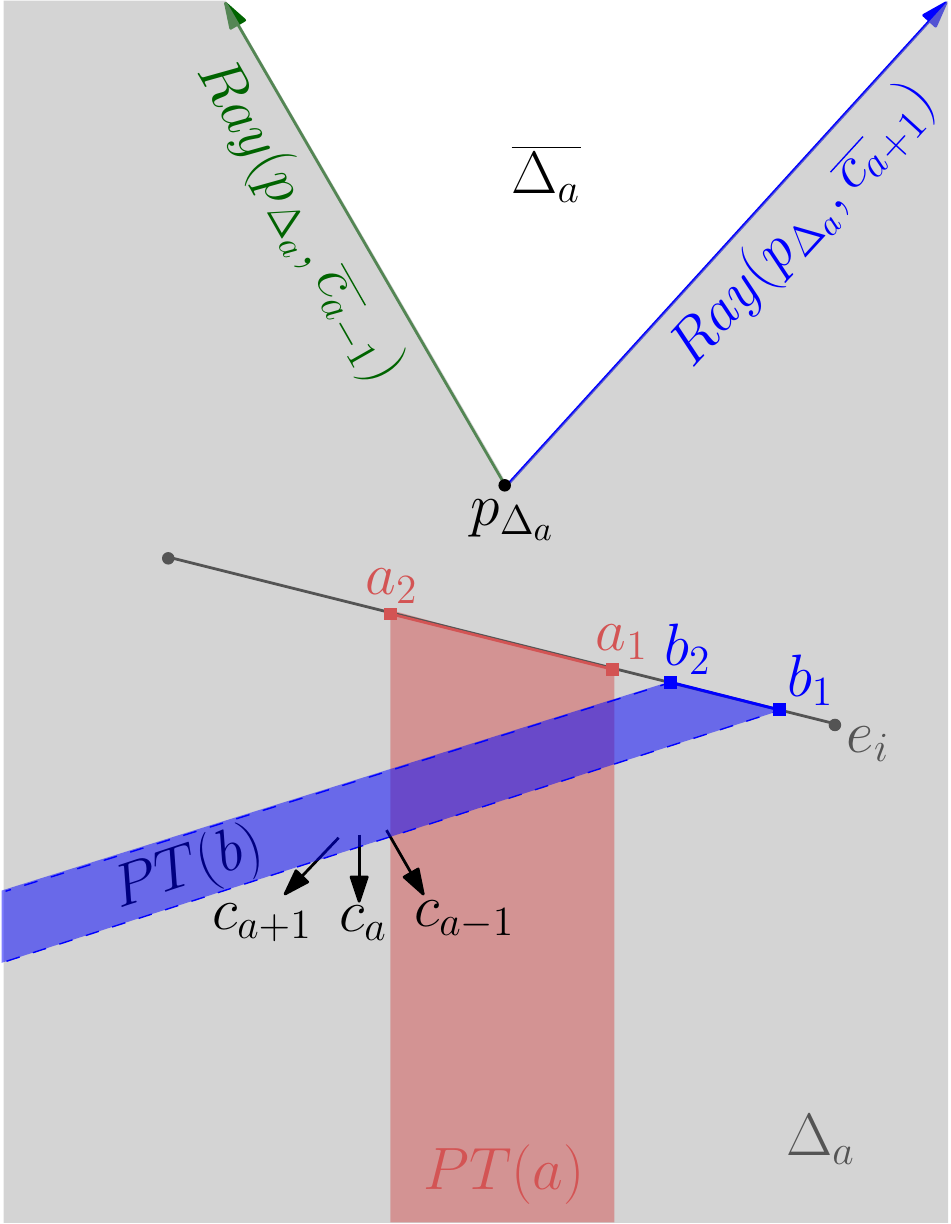}
  \caption{Case A.}
  \label{fig:delta_b_case_a}
\end{subfigure}%
 \hspace{2em}%
      \begin{subfigure}{.5\textwidth}
      \centering
  \includegraphics[scale=0.55]{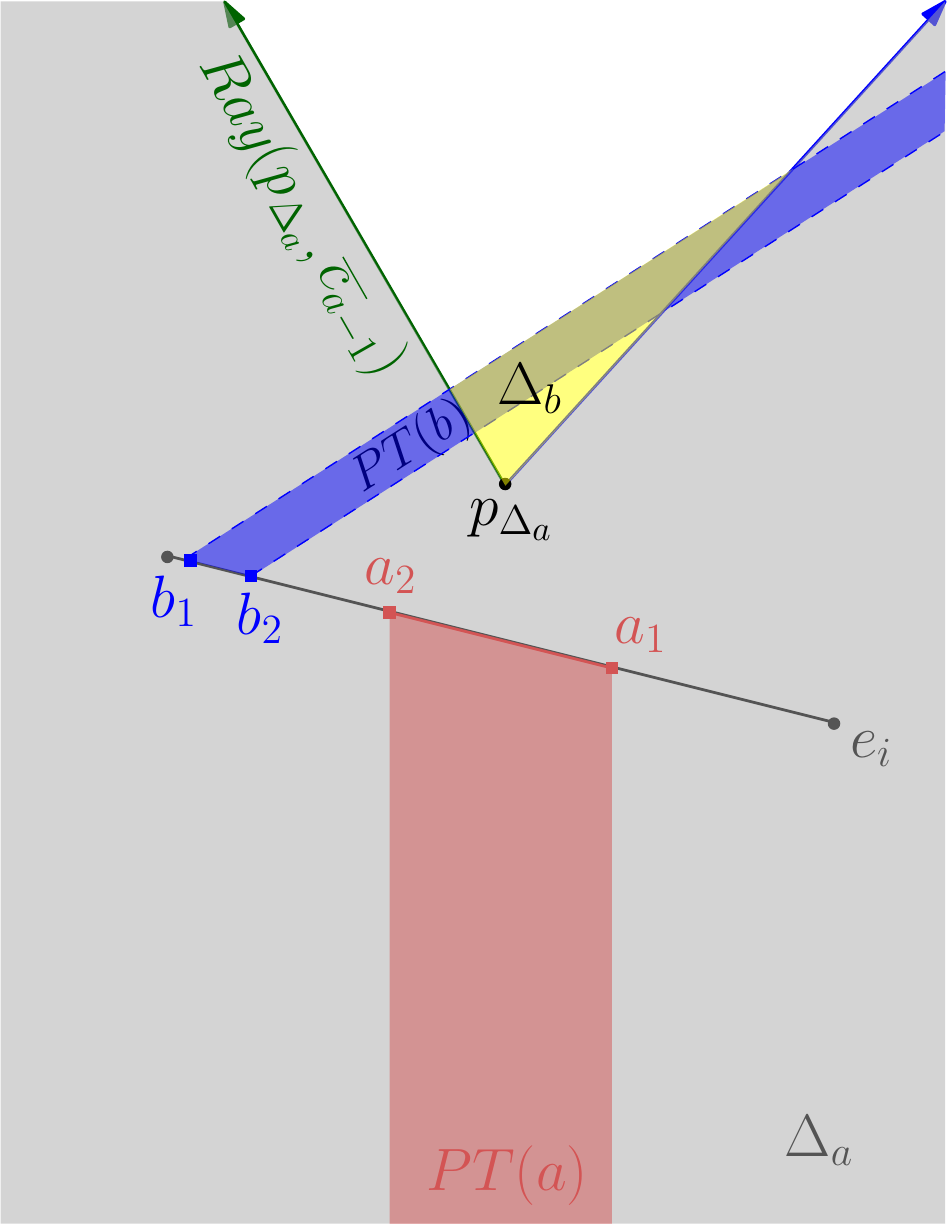}
  \caption{Case B.}
  \label{fig:delta_b_case_b}
\end{subfigure}%
\caption{$R_{a \cup b}$ in Cases A and B.}
\label{fig:delta_b}
\end{figure}

{\bf Case C: $PT(b) \not\subseteq \Delta_a$ and $c_b \in \{\overline{c_a}_{-1},\overline{c_a},\overline{c_a}_{+1}\}$} (see Figure~\ref{fig:delta_b_case_c}).\\
Assume, e.g., that $c_b=\overline{c_a}_{+1}$. Then, $c_j\in \{c_{a-1}, c_a\}$ (since $b\notin I^+(e_i, \overline{c_j})$). In this case, $R_{a\cup b}$ is the union of $\Delta_a$ and the `trapezoid' that is formed by the rays defining $\Delta_a$ and the appropriate side of $PT(b)$; see the yellow region $\Delta_b$ in Figure~\ref{fig:delta_b_case_c1}.
The subcases $c_b=\overline{c_a}_{-1}$ and $c_b=\overline{c_a}$ are similar and are depicted in Figure~\ref{fig:delta_b_case_c2} and Figures~\ref{fig:delta_b_case_c3} and~\ref{fig:delta_b_case_c4}, respectively. 

\begin{figure}[!h]
\begin{subfigure}{.5\textwidth}
\centering
  \includegraphics[scale=0.55]{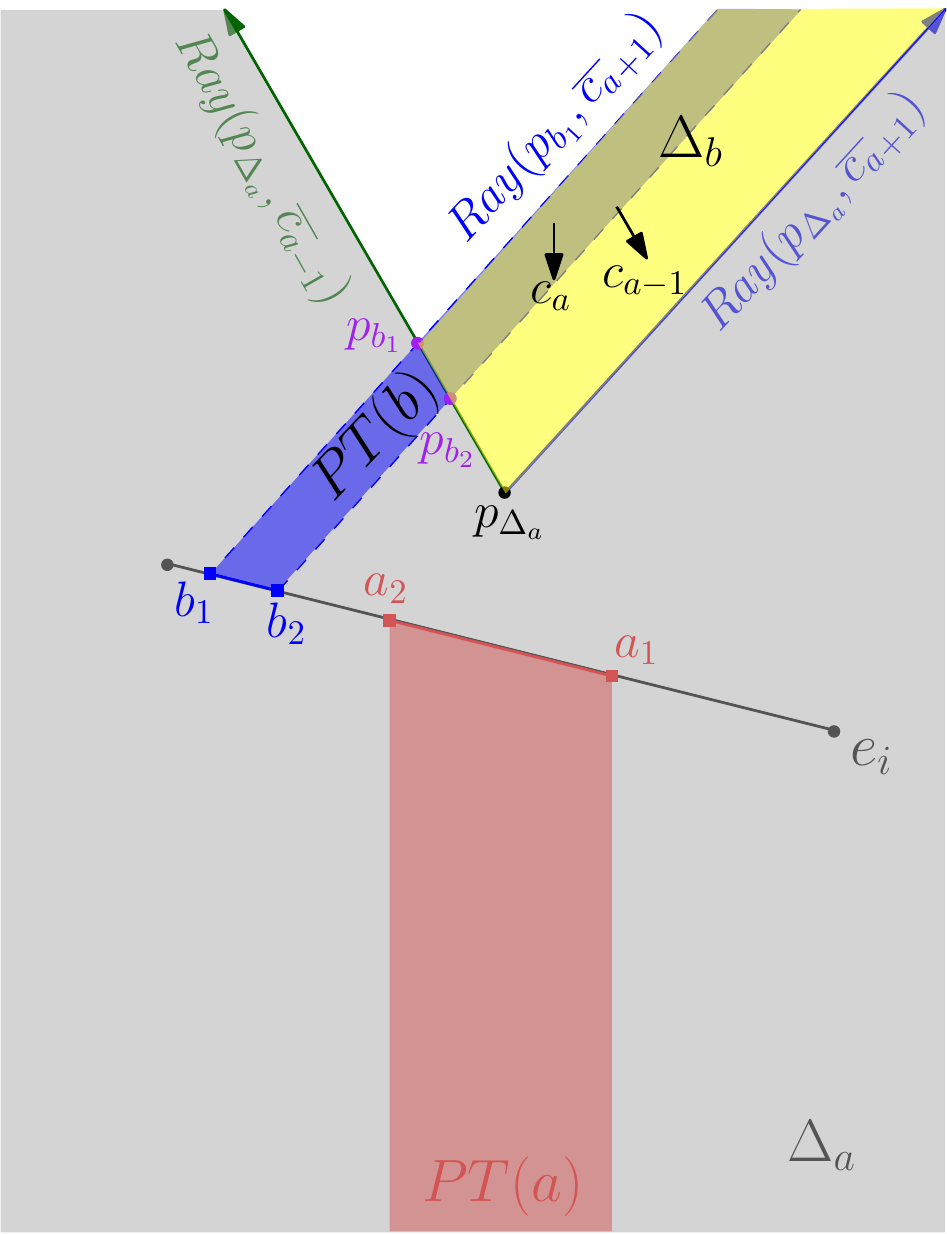}
  \caption{$c_b=\overline{c_a}_{+1}$ and $c_j\in \{c_{a-1}, c_a\}$.}
  \label{fig:delta_b_case_c1}
\end{subfigure}%
 \hspace{2em}%
      \begin{subfigure}{.5\textwidth}
      \centering
  \includegraphics[scale=0.55]{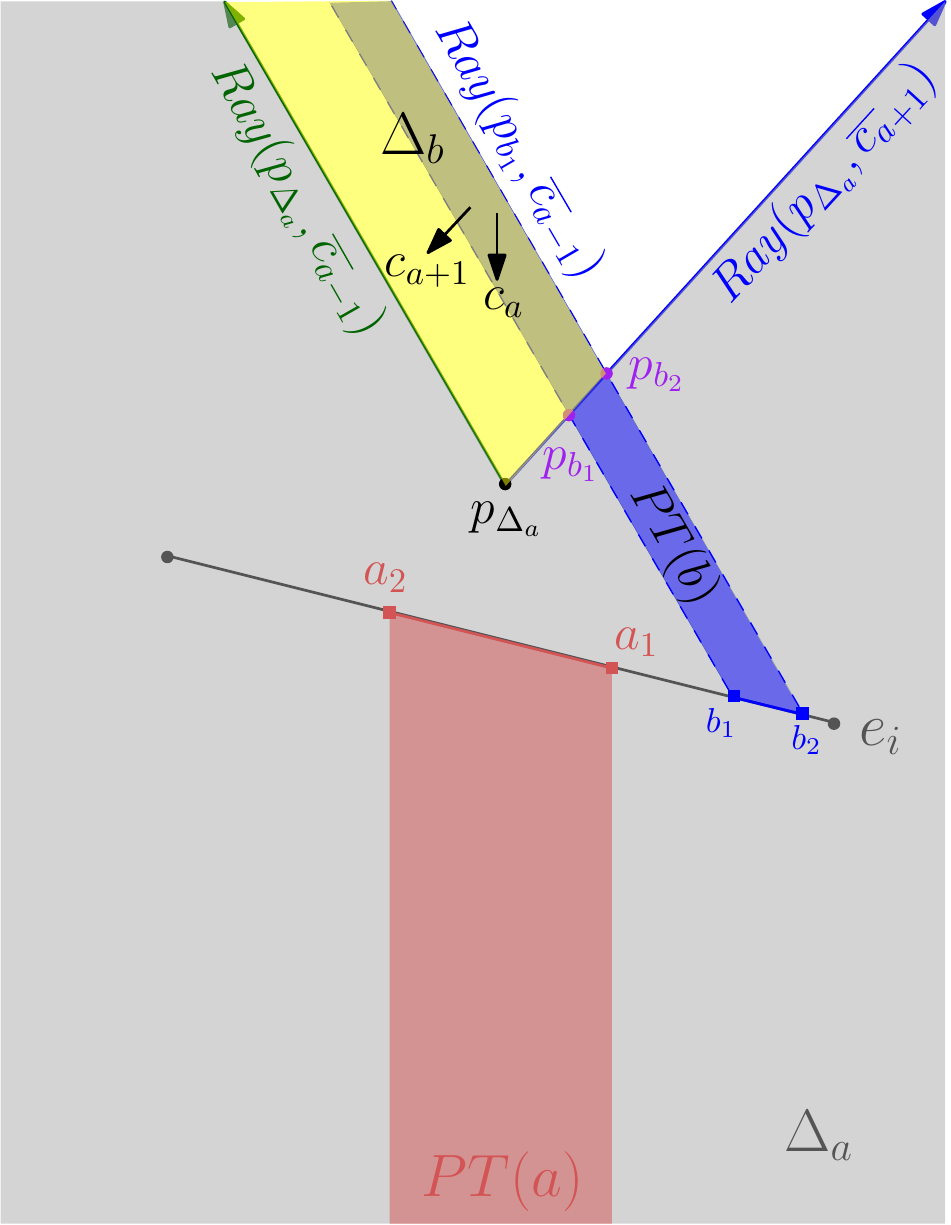}
  \caption{$c_b=\overline{c_a}_{-1}$ and $c_j\in \{c_{a+1}, c_a\}$.}
  \label{fig:delta_b_case_c2}
\end{subfigure}%
\newline
\newline
\newline
\begin{subfigure}{.5\textwidth}
\centering
  \includegraphics[scale=0.55]{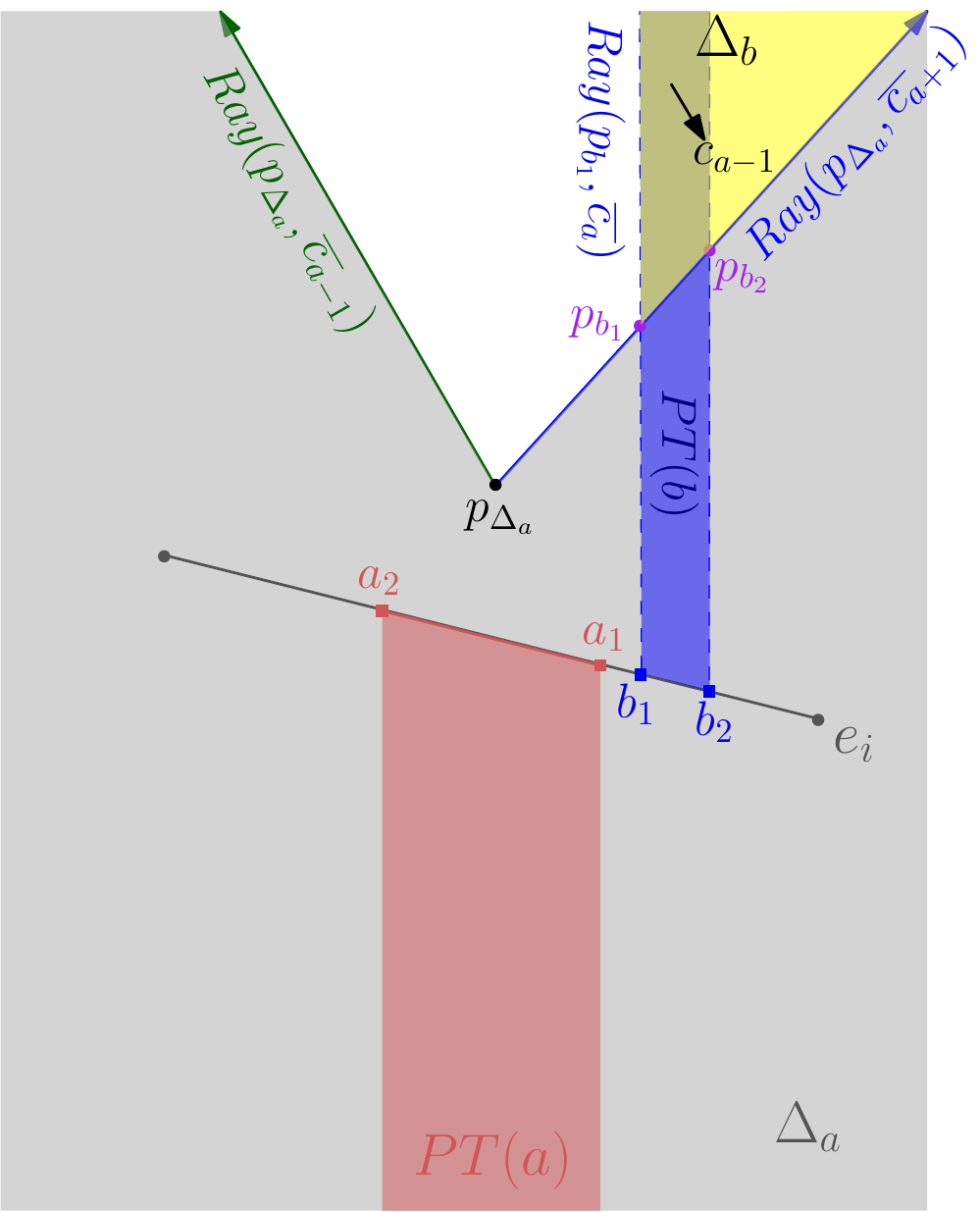}
  \caption{$c_b=\overline{c_a}$ and $c_j=c_{a-1}$.}
  \label{fig:delta_b_case_c3}
\end{subfigure}%
 \hspace{2em}%
      \begin{subfigure}{.5\textwidth}
      \centering
  \includegraphics[scale=0.55]{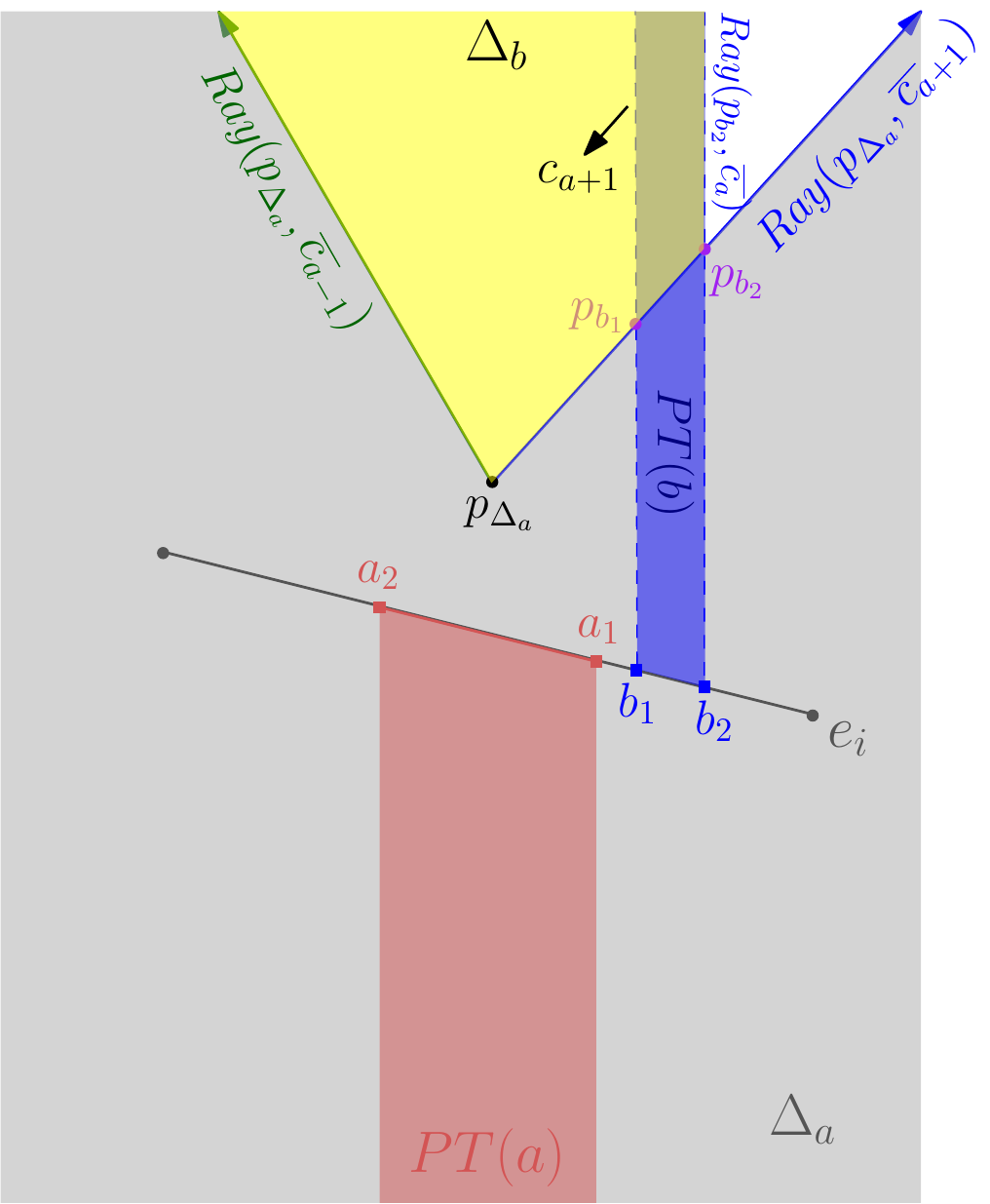}
  \caption{$c_b=\overline{c_a}$ and $c_j=c_{a+1}$.}
  \label{fig:delta_b_case_c4}
\end{subfigure}%
\caption{$R_{a \cup b}$ in Case C.}
\label{fig:delta_b_case_c}
\end{figure}

\begin{figure}[t]
\centering
  \includegraphics[scale=0.55]{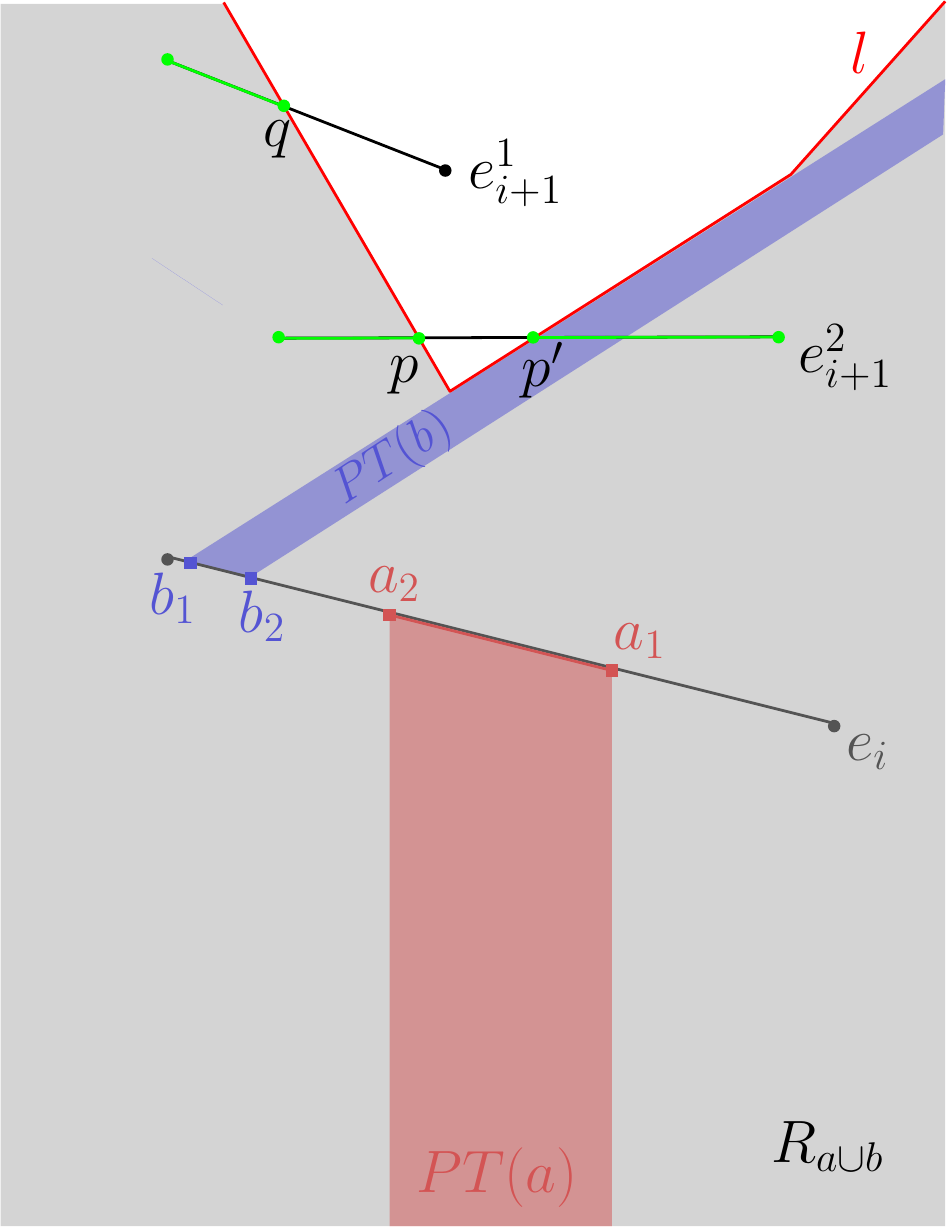}
  \caption{$e_{i+1}$ intersects $R_{a\cup b}$'s boundary either at a single point $q$ or at two points $p$ and $p'$.}
 \label{fig:e_intersect_R_aub}
\end{figure}

We conclude that in each of the cases A--C, $R_{a\cup b}$ is bounded by an infinite convex chain $l$ consisting of at most three edges (see the red chain in Figure~\ref{fig:e_intersect_R_aub}).
But this implies that, unless $R_{a\cup b}\cap e_{i+1}$ is $e_{i+1}$ or is empty, $R_{a\cup b}\cap e_{i+1}$ consists of at most two intervals, each of which contains an endpoint of $e_{i+1}$ (see edges $e^1_{i+1}$ and $e^2_{i+1}$ in Figure~\ref{fig:e_intersect_R_aub}). 

We have shown that by passing through the intervals $a$ and $b$, at most two intervals (with associated length $l(e_i)+2$ and orientation $c_j$) are formed on $e_{i+1}$. Moreover, each of these intervals contains an endpoint of $e_{i+1}$. Therefore, the total number of such intervals that are formed on $e_{i+1}$, by passing through the intervals in $I(e_i) \cup I^+(e_i) \setminus \{ I^+(e_i, \overline{c_j})\cup I^+(e_i,c_j) \}$ is at most two (for each endpoint $p$ of $e_{i+1}$, we retain only the longest interval with $p$ as one of its endpoints).

Finally, observe that by passing through an interval in $I^+(e_i,\overline{c_j})$ and turning backwards in orientation $c_j$, at most one interval is formed on $e_{i+1}$, which does not necessarily contain an endpoint of $e_{i+1}$.

We conclude that at most $|I^+(e_i,\overline{c_j})| + 2$ intervals (with associated length $l(e_i)+2$ and orientation $c_j$) are formed on $e_{i+1}$ during the execution of the algorithm. 
\end{proof}

\end{document}